\documentclass[12pt]{article}
\usepackage{graphicx}
\def\hybrid{\topmargin 0pt      \oddsidemargin 0pt
        \headheight 0pt \headsep 0pt
       \voffset-1cm
        \textwidth 6.25in       % A4 paper
       \textheight 9.5in       % A4 paper
        \marginparwidth 0.0in
        \parskip 5pt plus 1pt   \jot = 1.5ex}
\catcode`\@=11
\def\marginnote#1{}

\newcount\hour
\newcount\minute
\newtoks\amorpm
\hour=\time\divide\hour by60
\minute=\time{\multiply\hour by60 \global\advance\minute by-\hour}
\edef\standardtime{{\ifnum\hour<12 \global\amorpm={am}%
        \else\global\amorpm={pm}\advance\hour by-12 \fi
        \ifnum\hour=0 \hour=12 \fi
        \number\hour:\ifnum\minute<10 0\fi\number\minute\the\amorpm}}
\edef\militarytime{\number\hour:\ifnum\minute<10 0\fi\number\minute}

\def\draftlabel#1{{\@bsphack\if@filesw {\let\thepage\relax
   \xdef\@gtempa{\write\@auxout{\string
      \newlabel{#1}{{\@currentlabel}{\thepage}}}}}\@gtempa
   \if@nobreak \ifvmode\nobreak\fi\fi\fi\@esphack}
        \gdef\@eqnlabel{#1}}
\def\@eqnlabel{}
\def\@vacuum{}
\def\draftmarginnote#1{\marginpar{\raggedright\scriptsize\tt#1}}

\def\draftlabel#1{{\@bsphack\if@filesw {\let\thepage\relax
   \xdef\@gtempa{\write\@auxout{\string
      \newlabel{#1}{{\@currentlabel}{\thepage}}}}}\@gtempa
   \if@nobreak \ifvmode\nobreak\fi\fi\fi\@esphack}
        \gdef\@eqnlabel{#1}}
\def\@eqnlabel{}
\def\@vacuum{}
\def\draftmarginnote#1{\marginpar{\raggedright\scriptsize\tt#1}}

\def\draft{\oddsidemargin -.5truein
        \def\@oddfoot{\sl preliminary draft \hfil
        \rm\thepage\hfil\sl\today\quad\militarytime}
        \let\@evenfoot\@oddfoot \overfullrule 3pt
        \let\label=\draftlabel
        \let\marginnote=\draftmarginnote
   \def\@eqnnum{(\theequation)\rlap{\kern\marginparsep\tt\@eqnlabel}%
\global\let\@eqnlabel\@vacuum}  }

\def\numberbysection{\@addtoreset{equation}{section}
        \def\theequation{\thesection.\arabic{equation}}}

\def\underline#1{\relax\ifmmode\@@underline#1\else
        $\@@underline{\hbox{#1}}$\relax\fi}

\def\titlepage{\@restonecolfalse\if@twocolumn\@restonecoltrue\onecolumn
     \else \newpage \fi \thispagestyle{empty}\c@page\z@
        \def\thefootnote{\fnsymbol{footnote}} }

\def\endtitlepage{\if@restonecol\twocolumn \else  \fi
        \def\thefootnote{\arabic{footnote}}
        \setcounter{footnote}{0}}  %\c@footnote\z@ }
\relax

\numberbysection
\hybrid

\newfont{\Bbb}{msbm10 scaled 1\@ptsize00}
\newfont{\Bbbb}{msbm7 scaled 1\@ptsize00}
\newcommand{\CC}{\mbox{\Bbb C}}

\newcommand{\DDD}{\raise-1pt\hbox{$\mbox{\Bbbb D}$}}

\newcommand{\UUU}{\raise-1pt\hbox{$\mbox{\Bbbb U}$}}

\newcommand{\ZZ}{\mbox{\Bbb Z}}
\newcommand{\z}{\raise-1pt\hbox{$\mbox{\Bbbb Z}$}}

\newcommand{\sss}{\raise-1pt\hbox{$\mbox{\Bbbb S}$}}

\def\beq{\begin{equation}}
\def\eeq{\end{equation}}
\def\p{\partial}

\def\sn{{\sf sn}}
\def\cn{{\sf cn}}
\def\dn{{\sf dn}}

\newtheorem{theorem}{Theorem}[section]

\newtheorem{lemma-definition}{Lemma-Definition}[section]
\newtheorem{corollary}{Corollary}[section]

\newtheorem{remark}{Remark}[section]

\newtheorem{proposition}{Proposition}[section]

\def\square{\hfill
{\vrule height6pt width6pt depth1pt} \break \vspace{.01cm}}

\begin{document}

\begin{titlepage}

\title{Dispersionless modified DKP hierarchy as 
the Yang-Baxter equation}

\author{
A.~Zabrodin\thanks{
National Research University Higher School of Economics,
20 Myasnitskaya Ulitsa,
Moscow 101000, Russia and
NRC ``Kurchatov institute'', Moscow, Russia;
e-mail: zabrodin@itep.ru}}

\date{July 2026}
\maketitle

\vspace{-7cm} \centerline{ \hfill ITEP-TH-30/26}\vspace{7cm}

\begin{abstract}

We show that the dispersionless version of the modified DKP hierarchy originally defined as the limit of relations for the tau-function of 
the Hirota-Miwa type has an equivalent
reformulation as the Yang-Baxter equation for Baxter's $R$-matrix
of Boltzmann weights for the 8-vertex model.

\end{abstract}

\end{titlepage}

\vspace{5mm}

%

%\newpage
\tableofcontents

\vspace{5mm}

\section{Introduction}

\subsection*{Instead of an epigraph}

It seems appropriate to start with a remembrance. At a conference 
in the late 80s of the last century, 
Volodya Drinfeld gave a talk and wrote 
a rather complicated-looking equation on the blackboard. 
From the audience, he was asked a question: ``But how are you 
going to solve this equation?'' Volodya gave an answer, 
which could be taken as an epigraph to the present article:
``We never {\it solve} equations. This is not our task. 
We only prove that some equations are equivalent to some others''.

\subsection{The main result from a bird's eye view}

The purpose of this article is to connect integrable 
systems that are, as it were, at opposite poles of the theory: 
classical nonlinear partial differential equations and
quantum spin chains or closely related vertex models of 
statistical mechanics.
Let's formulate the main result in a rather informal way, 
avoiding any details (some of which, however, are important):

\begin{itemize}
\item
Integrable hierarchies of partial differential equations (such as
DKP, large BKP and KP) in the limit of zero dispersion
have an equivalent formulation 
in the form of the Yang-Baxter equation for quantum $R$-matrices.
In the case of DKP the $R$-matrix is the Baxter $4\times 4$ matrix
of Boltzmann's weights for the 8-vertex model with elliptic
dependence on spectral parameter. For the other two hierarchies,
the $R$-matrix is a degeneration of the Baxter one, when the
modular parameter $\tau$ tends to $+i0$ (for large BKP) or
to $+i\infty$ (for KP).
\end{itemize}

\noindent
The existence of such a connection seems absolutely unexpected
and, therefore, this statement sounds somewhat defiant.
Indeed, it is difficult to find integrable systems 
that are more distant from each other and different in nature.
In this paper we consider only the general case when all four
Boltzmann weights $a,b,c,d$ are not identically zero. It corresponds
to the DKP hierarchy. The degenerate cases will be discussed
elsewhere. 

A more detailed summary of the contents of the paper is presented
in the next subsection. In order not to interrupt 
the presentation, we do not provide any references to the literature
there. The necessary references are given below in Section 
\ref{section:background}.

\subsection{A more detailed summary}

First of all, we should clarify that we are not dealing literally 
with the DKP hierarchy, but with its ``modified''
version, mDKP. The term ``modified'' is 
understood in the following sense.
Any integrable hierarchy 
(call it ${\cal H}$) has at least one
infinite set of independent variables (times), usually denoted
as ${\bf t}=\{t_1, t_2, t_3, \ldots \}$ supplemented by a 
discrete variable $n\in \ZZ$. Such a hierarchy can be naturally
embedded into a bigger one which is commonly called its 
multi-component extension, ${\cal H}^{\rm (multi )}$. 
Each component is associated with
its own set of times, 
${\bf t}_{\alpha}=\{t_{\alpha ,1}, t_{\alpha ,2}, 
t_{\alpha ,3}, \ldots \}$ and a discrete variable $n_{\alpha}$,
where $\alpha \in \{1, \ldots , N\}$. The embedding 
${\cal H} \subset {\cal H}^{\rm (multi )}$ means
that one simply identifies ${\bf t}$ and $n$ with, say, 
${\bf t}_1$ and $n_1$ in ${\cal H}^{\rm (multi )}$. 
In these terms, the {\it modified} hierarchy
${\cal H}$ (call it $m{\cal H}$) is a minimal 
possible extension of ${\cal H}$ by switching
just one additional discrete variable (say, $n_2$) and freezing all
other variables ${\bf t}_{\alpha}$ with $\alpha \geq 2$ and 
$n_{\alpha}$ with $\alpha \geq 3$. This is the way how we define 
the mDKP hierarchy. To the best of our knowledge,
it was not considered in the literature before. 
The relationship DKP--mDKP is similar to that of the
KP hierarchy and its well known modified version (mKP).

So, the mDKP hierarchy we are dealing with has
the infinite set ${\bf t}=\{t_1, t_2, t_3, \ldots \}$ of continuous
times and two discrete variables $n, \bar n \in \ZZ$. We work
in the bilinear formalism, in which the main hero is the 
tau-function $\tau (n, \bar n, {\bf t})$ which serves as 
a universal dependent variable. It satisfies an infinite
number of bilinear differential and/or difference equations
which can be unified in one generating integral bilinear
relation. 

In this paper we are mainly interested in the so-called 
{\it dispersionless version} of the hierarchy. 
The passage to the dispersionless limit 
consists in introducing a dispersion parameter $\hbar$ and
the re-scaling $t_{k}\to t_{k}/\hbar$ for all $k \geq 1$,
$n\to t_0/\hbar$, $\bar n\to \bar t_0/\hbar$,
with the subsequent limit $\hbar \to 0$. The tau-function itself
does not exist in this limit.
However, a scaling limit of its logarithm still makes sense.
More precisely, introduce the function 
\beq\label{kp6a}
F(t_0, \bar t_0, {\bf t}; \hbar )=\hbar^2 \log \left [
\tau \Bigl (\hbar^{-1}t_0, \hbar^{-1}
\bar t_0,\hbar^{-1}{\bf t}  \Bigr )
\right ]
\eeq
and consider the limit
$
\displaystyle{F(t_0, \bar t_0, {\bf t})=
\lim\limits_{\hbar \to 0} 
F(t_0, \bar t_0, {\bf t}; \hbar )}
$, which is known to exist for sufficiently broad class
of solutions (for example, solutions related to models
of random matrices).
The $F$-function satisfies 
an infinite number of compatible 
highly nonlinear differential equations which 
form the {\it dispersionless hierarchy} and can be
obtained from the bilinear equations for the 
tau-function. They can be naturally written in terms
of the differential operators
\beq\label{int2}
\nabla (z)=\p_{t_0} +\sum_{k\geq 1}\frac{z^{-k}}{k}\, \p_{t_k},
\qquad
\p_0 =\p_{t_0}, \qquad \bar \p_0 =\p_{\bar t_0}
\eeq
acting to the $F$-function.
The main message of this paper is that these equations are
equivalent to the Yang-Baxter equation for the Baxter 
$4\times 4$ quantum
$R$-matrix
\beq\label{R10a}
{\sf R}=\left (
\begin{array}{cccc}
a & 0 & 0& d 
\\ 
0 & b & c & 0
\\ 
0 & c & b & 0
\\ 
d & 0 & 0& a
\end{array}
\right ) ,
\eeq
where $a,b,c,d$ are Boltzmann's weights for the symmetric 
8-vertex model. (This $R$-matrix plays also a crucial role
in Baxter's solution of the quantum $XYZ$ spin chain.)
More precisely, we prove the following theorem.

\begin{theorem}\label{theorem:main}
Set
\beq\label{int3}
\begin{array}{l}
g(z, \zeta )=(z^{-1} -\zeta^{-1})e^{\nabla (z)\nabla (\zeta )F},
\\ \\
w(z)=z^{-1}e^{\nabla (z) \p_0 F}, \quad
v(z)=e^{\nabla (z)\bar \p_0 F},
\quad R=e^{\p_0 \bar \p_0 F},
\\ \\
\displaystyle{
f(z, \zeta )=v(\zeta )\, \frac{Rw(\zeta )+
v(z)g(z, \zeta )}{v(z)w(\zeta )+
Rg(z, \zeta )}.
}
\end{array}
\eeq
Then the Yang-Baxter equation
\beq\label{YB0}
{\sf R}_{12}(z_1, z_2) {\sf R}_{13}(z_1, z_3) {\sf R}_{23}(z_2, z_3)=
{\sf R}_{23}(z_2, z_3) {\sf R}_{13}(z_1, z_3) {\sf R}_{12}(z_1, z_2)
\eeq
for the $R$-matrix (\ref{R10a}) in which
$$
a=f(z, \zeta ), \quad b=g(z, \zeta ), \quad c=R, \quad
d=Rg(z, \zeta )f(z, \zeta )
$$
is equivalent to the dispersionless mDKP hierarchy.
\end{theorem}

In this form, the theorem describes the case
of ``general position'' when all elements $a,b,c,d$ are not 
identically zero.
This case corresponds to the 8-vertex model (and the $XYZ$ spin chain)
and can be parametrized by means of elliptic functions with 
some modular parameter $\tau$. It turns out that the degenerate
cases $\tau \to +i0$ and $\tau \to +i\infty$, when parametrization
via trigonometric or hyperbolic functions is possible, correspond,
in the sense of Theorem \ref{theorem:main},
to the dispersionless modified large BKP and mKP hierarchies respectively.
However, these degenerate cases are not considered in this paper
in detail.

At first glance, the statement of the theorem 
seems rather surprising since among all the variety of integrable 
systems, there are hardly two more distant from each 
other than classical dispersionless equations and 
quantum spin chains. The secret here is that 
behind both systems there is a certain elliptic curve (or its rational degenerations in the limiting cases). On the side of spin chains and
vertex models it is the {\it spectral curve}, on which the spectral
parameter of $R$-matrices lives. 
On the side of dispersionless hierarchies, it is
the so-called {\it dynamical curve} defined by the following polynomial 
relation between the functions $w(z)$ and $v(z)$:
\beq\label{YB3a}
R^2 \Bigl (v^2(z)w^2(z)+1 \Bigr )
-v^2(z)-w^2(z)+ V v(z)w(z)=0,
\eeq
where
$V=2\, e^{\p_0 (\bar \p_0 -\p_0)F}\p_{t_1}\bar \p_0 F$. In general
it is a smooth elliptic curve isomorphic to the spectral curve
that naturally arises in solving the Yang-Baxter equation.
The complete proof of Theorem \ref{theorem:main} 
requires a special change of dynamical variables, 
which is motivated by 
uniformization of the dynamical curve using elliptic functions.

\subsection{Background of the issue and references}
\label{section:background}

Before we proceed to further details, let’s comment on the background 
of the issue, and also
give some necessary references.

Regarding the side of quantum $R$-matrices and 
the Yang-Baxter equation, the main reference is Baxter's book
\cite{Baxter}. His original papers \cite{Baxter71,Baxter72,Baxter73}
may be also useful. The fundamental role of the Yang-Baxter equation
in the theory of quantum integrable systems was revealed by
Faddeev and Takhtajan in \cite{FT79}, see also books
\cite{Gaudin-book}, \cite{KBI93} and \cite{Slavnov-book}.

Our approach to classical integrable hierarchies of nonlinear
differential and difference equations is based on the works
\cite{DJKM83,JM83} of the Kyoto school. The main object is the
tau-function which serves as a universal dependent variable and
satisfies an infinite set of bilinear equations. 
See also the book \cite{HBbook} and references therein.

The most popular and well-studied integrable hierarchies are
the Ka\-dom\-tsev-Pet\-vi\-ash\-vi\-li (KP) and Toda lattice hierarchies
\cite{DJKM83,UT84} and their multi-component 
generalizations \cite{DJKM81}--\cite{TZ25}. 
The DKP hierarchy is less studied.
It is one of the integrable hierarchies with 
$D_{\infty}$ symmetries introduced by 
Jimbo and Miwa in 1983 \cite{JM83}.
It was subsequently rediscovered several times and came to be
also known as the coupled KP hierarchy \cite{HO} and
the Pfaff lattice \cite{AHM,ASM}, see also
\cite{Kakei,IWS,Willox}. The latter name is motivated by the fact 
that some solutions to the hierarchy are expressed through Pfaffians. 
The solutions and the algebraic structure were studied in 
\cite{Kodama}, the relation to matrix integrals 
was elaborated in \cite{AHM,ASM,Kakei,Vandeleur,Orlov}.
Bearing certain similarities with the KP and Toda lattice hierarchies,
the DKP one is essentially different and not that well studied.
Recently, in our papers \cite{SZ24,SZ25a} multi-component generalizations
of these hierarchies were introduced.

The general approach to hierarchies with zero dispersion
which we follow in this paper was developed by Takasaki and 
Takebe in \cite{TT95}. This approach was extended to the
hierarchies of Pfaff type (DKP and Pfaff-Toda) 
in \cite{Takasaki07,Takasaki09}. 

An important observation made by Takasaki in
\cite{Takasaki07,Takasaki09} is that the dispersionless versions
of integrable hierarchies of the Pfaff type such as DKP and Pfaff-Toda
contain a hidden (in general smooth) {\it elliptic 
curve} naturally built in the structure
of the hierarchy. This observation gave a special additional 
interest to the theory and had significant consequences.

The theory of dispersionless Pfaff-type hierarchies 
was further developed in \cite{AZ14,AZ15}, where it was
shown that a special change of dynamical 
variables motivated by uniformization
of the elliptic curve via elliptic functions, makes the structure
of the hierarchy especially nice and clear 
(at the cost of having to deal with 
non-elementary functions). In
\cite{SZ24,SZ25b} this approach was generalized to the
multi-component hierarchies of Pfaff type. The result was that
in this more general case the same elliptic curve emerges, but 
for multi-component hierarchies it is equipped with more than one
marked points (which are dynamical variables).
In our recent paper \cite{SZ26a}
we have suggested to call it the {\it dynamical curve},
since its parameters depend on the hierarchical times.
Moreover, in \cite{Z24} it was shown that the dynamical curve
can also be detected in hierarchies of the KP type (with the number
of component more than 1), but in this case it is a smooth {\it rational}
curve (of genus 0) rather than elliptic one.

\subsection{The notation and conventions related to quantum\\
$R$-matrices}

We will need some basic concepts and standard designations adopted in the theory of quantum integrable systems and vertex models.
The main object for us is a quantum $R$-matrix satisfying the Yang-Baxter equation. In this paper, we deal with the simplest example of it
known as the Baxter $R$-matrix. In the context of vertex models, 
it is the matrix of Boltzmann's weights associated with a vertex
of the 8-vertex model.

Let $V\cong \CC^2$ be the 2-dimensional linear space over $\CC$.
Let $({\sf e}_1, {\sf e}_2)$ be a basis in $V$. 
Given a linear operator ${\sf R}\in {\rm End}(V\otimes V)$, we
represent it as a $4\times 4$ matrix in the natural basis
$({\sf e}_1 \otimes {\sf e}_1, \, {\sf e}_1 \otimes {\sf e}_2,\,
{\sf e}_2 \otimes {\sf e}_1, \, {\sf e}_2 \otimes {\sf e}_2)$
in $V\otimes V$. In particular, we will be mostly interested
in symmetric matrices of the form (\ref{R10a}).
Such a matrix can be also represented as
\beq\label{R1b}
{\sf R}=\sum_{j=0}^3 w_j \sigma_j \otimes \sigma_j,
\eeq
where $\sigma_j$ are Pauli matrices
$$
\sigma_0 =\left (\begin{array}{rr} 1 &0 \\ 0 & 1 \end{array}\right ),
\quad
\sigma_1 =\left (\begin{array}{rr} 0 &1 \\ 1 & 0 \end{array}\right ),
\quad
\sigma_2 =\left (\begin{array}{rr} 0 &-i \\ i & 0 \end{array}\right ),
\quad
\sigma_3 =\left (\begin{array}{rr} 1 &0 \\ 0 & -1 \end{array}\right ),
$$
and
$$
\begin{array}{l}
w_0=\frac{1}{2}(a+b), \quad w_3=\frac{1}{2}(a-b), 
\quad w_1=\frac{1}{2}(c+d), \quad w_2=\frac{1}{2}(c-d).
\end{array}
$$

Further, 
let $V_1, V_2, V_3$ be three copies of the linear space $V\cong \CC^2$,
each equipped with the basis as before.
By ${\sf R}_{12}\in {\rm End}(V_1\otimes V_2)$ 
we denote the matrix that acts in the 8-dimensional space
$V_1\otimes V_2 \otimes V_3$ in the following way:
it acts non-trivially (as ${\sf R}$) in $V_1\otimes V_2$ 
and trivially in $V_3$. Similarly,
${\sf R}_{13}$ and ${\sf R}_{23}$
act as ${\sf R}$ in $V_1\otimes V_3$ and $V_2\otimes V_3$ respectively
and trivially in the third space. 
Let 
${\sf P}_{ab}\in {\rm End}(V_a \otimes V_b)$ be the permutation 
operator acting as ${\sf P}({\sf e}_i \otimes {\sf e}_j)=
{\sf e}_j \otimes {\sf e}_i$, then ${\sf R}_{ba}=
{\sf P}_{ab}{\sf R}_{ab}{\sf P}_{ab}$. In the matrix form,
\beq\label{P1}
{\sf P}=\left (
\begin{array}{cccc}
1 & 0 & 0& 0 
\\ 
0 & 0 & 1 & 0
\\ 
0 & 1 & 0 & 0
\\ 
0 & 0 & 0& 1
\end{array}
\right ) .
\eeq
Note that for symmetric
matrices of the form (\ref{R10a}) it holds ${\sf R}_{12}={\sf R}_{21}$.

The $R$-matrix is required to satisfy the Yang-Baxter equation:
\beq\label{RRR6b}
{\sf R}_{12}{\sf R}'_{13} {\sf R}''_{23}=
{\sf R}''_{23} {\sf R}'_{13} {\sf R}_{12}.
\eeq
The Baxter $R$-matrix is a solution to this equation under the
assumption that each $R$-matrix has the structure (\ref{R10a}), i.e.,
nonzero elements ($a', b', c', d'$ for ${\sf R}'$ and
$a'', b'', c'', d''$ for ${\sf R}''$) 
are at the same places as in (\ref{R10a}).

Concepts and designations specific to the theory of 
integrable hierarchies will be introduced in the main text 
at the appropriate places.

\subsection{Structure of the paper}

In Section 2 the mDKP hierarchy is introduced. Its definition 
looks most natural in the context 
of the more general Pfaff-Toda hierarchy introduced by Takasaki
in \cite{Takasaki09}. The full set of bilinear equations for
the tau-function is obtained. Section 3 is devoted to the 
dispersionless limit of the mDKP hierarchy. First, in Section 3.1,
we give its ``algebraic formulation'' which consists in finding
dispersionless limits of the bilinear equations obtained in
Section 2. Next, in Section 3.2, the ``Yang-Baxter formulation''
is given, which is outlined above. Its equivalence to the whole
hierarchy is proved in the next Section 4. The proof is based 
on a special change of variables motivated by uniformization
of the dynamical elliptic curve (or, in the context
of $R$-matrices, the spectral curve) using
elliptic functions. 
In Section 5 some perspectives and problems for future research
are discussed. There are also two appendices. Appendix A contains
the necessary information about theta-functions and elliptic
functions. Appendix B provides some details 
of the uniformization of algebraic curves, elliptic and rational.

\section{Modified DKP hierarchy}
\label{section:mDKP}

It is natural to introduce the modified DKP hierarchy
starting from the more general Pfaff-Toda hierarchy \cite{Takasaki09}.
Independent variables of the latter are two infinite sets
of continuous times
$$
{\bf t}=\{t_1, t_2, t_3, \ldots \}, \qquad
\bar {\bf t}=\{\bar t_1, \bar t_2, \bar t_3, \ldots \},
$$
and two discrete variables $n, \bar n \in \ZZ$ such that
$n-\bar n \in 2 \ZZ$ (the bar does not 
mean complex conjugation). In the framework of the bilinear
formalism, the universal dependent variable is the tau-function
$\tau (n, \bar n, {\bf t}, \bar {\bf t})$ satisfying the general
bilinear equation
\beq\label{bil1}
\begin{array}{l}
\displaystyle{ \phantom{+}
\oint_{C_{\infty}}\frac{dz}{z^2}\, z^{n-n'} e^{\xi ({\bf t}-{\bf t'}, z)}
\tau \bigl (n-1, \bar n, {\bf t}-[z^{-1}], \bar {\bf t}\bigr )
\tau \bigl (n'+1, \bar n', {\bf t'}+[z^{-1}],\bar {\bf t}\bigr )}
\\ \\
\displaystyle{
+ \, 
\oint_{C_{\infty}}\frac{dz}{z^2}\, z^{n'-n} e^{-\xi ({\bf t}-{\bf t'}, z)}
\tau \bigl (n+1, \bar n, {\bf t}+[z^{-1}], \bar {\bf t}\bigr )
\tau \bigl (n'-1, \bar n', {\bf t'}-[z^{-1}], \bar {\bf t}\bigr )}
\\ \\
\displaystyle{
= \, 
\oint_{C_{\infty}}\frac{dz}{z^2}\, 
z^{\bar n- \bar n'} e^{\xi (\bar {\bf t}-\bar {\bf t'}, z)}
\tau \bigl (n, \bar n -1, {\bf t}, \bar {\bf t}-[z^{-1}]\bigr )
\tau \bigl (n', \bar n'+1, {\bf t'}, \bar {\bf t'}+[z^{-1}]\bigr )}
\\ \\
\displaystyle{
+ \, 
\oint_{C_{\infty}}\frac{dz}{z^2}\, 
z^{\bar n'- \bar n} e^{-\xi (\bar {\bf t}-\bar {\bf t'}, z)}
\tau \bigl (n, \bar n +1, {\bf t}, \bar {\bf t}+[z^{-1}]\bigr )
\tau \bigl (n', \bar n'-1, {\bf t'}, \bar {\bf t'}-[z^{-1}]\bigr )},
\end{array}
\eeq
which is valid for all ${\bf t}, \bar {\bf t}, {\bf t'}, \bar {\bf t'}$
and $n,\bar n, n' ,\bar n'$ such that 
\beq\label{parity}
n-\bar n\in 2\ZZ +1, \qquad
n'-\bar n'\in 2\ZZ +1. 
\eeq
In (\ref{bil1}) the following standard
notations are used:
\beq\label{st}
\begin{array}{l}
\displaystyle{
\xi ({\bf t}, z)=\sum_{k\geq 1}t_k z^k,}
\\ \\
{\bf t}\pm [z^{-1}] =\{ t_1\pm z^{-1}, t_2\pm \frac{1}{2}\, z^{-2},
t_3\pm \frac{1}{3}\, z^{-3}, \ldots \}.
\end{array}
\eeq
The contour $C_{\infty}$ is a big circle of radius $R\to \infty$.

Equation (\ref{bil1}) is rather general and contains some 
important subhierarchies. For example,
if one puts $\bar n' =\bar n$ and $\bar {\bf t'}=\bar {\bf t}$,
the right-hand side of (\ref{bil1}) vanishes and the rest 
gives the bilinear equation for the DKP hierarchy: 
\beq\label{bil2}
\begin{array}{l}
\displaystyle{ \phantom{+}
\oint_{C_{\infty}}\frac{dz}{z^2}\, z^{n-n'} e^{\xi ({\bf t}-{\bf t'}, z)}
\tau \bigl (n-1, {\bf t}-[z^{-1}]\bigr )
\tau \bigl (n'+1, {\bf t'}+[z^{-1}]\bigr )}
\\ \\
\displaystyle{
+ \, 
\oint_{C_{\infty}}\frac{dz}{z^2}\, z^{n'-n} e^{-\xi ({\bf t}-{\bf t'}, z)}
\tau \bigl (n+1, {\bf t}+[z^{-1}], \bigr )
\tau \bigl (n'-1, {\bf t'}-[z^{-1}]\bigr )=0.}
\end{array}
\eeq
The variables $\bar n, \bar {\bf t}$ in the tau-function
$\tau (n, {\bf t})=\tau \bigl (n, \bar n, {\bf t}, 
\bar {\bf t}\bigr )$ do not participate in this
equation and enter as parameters.
It is assumed in (\ref{bil2}) that $n-n'\in 2 \ZZ$.

Another example (which is dealt with in this paper) is obtained 
if one puts $\bar {\bf t'}=\bar {\bf t}$ and $\bar n'=\bar n +1$.
In this case equation (\ref{bil1}) yields:
\beq\label{bil3}
\begin{array}{l}
\displaystyle{ \phantom{+}
\frac{1}{2\pi i}
\oint_{C_{\infty}}\frac{dz}{z^2}\, z^{n-n'} e^{\xi ({\bf t}-{\bf t'}, z)}
\tau \bigl (n-1, \bar n, {\bf t}-[z^{-1}]\bigr )
\tau \bigl (n'+1, \bar n +1, {\bf t'}+[z^{-1}]\bigr )}
\\ \\
\displaystyle{
+ \, 
\frac{1}{2\pi i}
\oint_{C_{\infty}}\frac{dz}{z^2}\, z^{n'-n} e^{-\xi ({\bf t}-{\bf t'}, z)}
\tau \bigl (n+1, \bar n, {\bf t}+[z^{-1}], \bigr )
\tau \bigl (n'-1, \bar n +1,{\bf t'}-[z^{-1}]\bigr )}
\\ \\
\displaystyle{
=\, \tau \bigl (n, \bar n +1,{\bf t}\bigr )
\tau \bigl (n', \bar n,{\bf t'}\bigr ).}
\end{array}
\eeq
Here the tau-function 
$\tau (n, \bar n, {\bf t})=\tau \bigl (n, \bar n, {\bf t}, 
\bar {\bf t}\bigr )$ is a function of the continuous times
${\bf t}$ and of the two discrete variables $n, \bar n$ while
the bar-times enter as parameters. Following the analogy with 
the interrelation between the KP and modified KP hierarchies it
is natural to call the hierarchy defined by equation (\ref{bil3})
{\it modified DKP hierarchy} (mDKP). Specifying the parity conditions
(\ref{parity}) to this case, we have $n'-\bar n \in 2\ZZ$ and
$n-\bar n \in 2\ZZ +1$, so $n'-n \in 2\ZZ +1$ in (\ref{bil3}).

The ``Miwa substitution'' \cite{Miwa82}
\beq\label{mkp2a}
\left \{\begin{array}{l}
\displaystyle{
n - n' = P^{+} - P^{-},}
\\ \\
\displaystyle{
{\bf t} - {\bf t}' =
\sum_{i = 1}^{P^{+}}[a_i^{-1}]
-\sum_{k = 1}^{P^{-}}[b_k^{-1}],}
\end{array}\right.
\eeq
where the points $a_i, b_k \in \CC$ are assumed to be distinct,
makes it possible to evaluate the integrals in (\ref{bil3}). 
Indeed, we have:
$$
z^{P^+ -P^-} e^{\xi ({\bf t}-{\bf t'}, z)} =
\prod_{i=1}^{P^+} \left (z^{-1}-a_i^{-1}\right )^{-1}
\prod_{j=1}^{P^-} \left (z^{-1}-b_j^{-1}\right ),
$$
and the integrals are given by sum of residues at the simple
poles at $a_i, b_k$. The special tuning of the shifts of discrete and continuous variables in (\ref{mkp2a}) guaranties that 
possible residues at $\infty$ vanish. 
The result is (for details of the derivation see \cite{SZ26a}):
\beq\label{bil4}
\begin{array}{l}
\displaystyle{
\sum_{s = 1}^{P^{+}} \prod_{i = 1, \neq s}^{P^{+}} E^{-1}(a_{s}, a_{i}) \,
\prod_{k = 1}^{P^{-}} E(a_{s}, b_{k})\,
\tau \Bigl( 
n + P^{+} \! -\! 1, \bar n, {\bf t} + \sum_{i\neq s}^{P^{+}}
[a_{i}^{-1}] \Bigr)}
\\ \\ \phantom{aaaaaaaaaaaaaaaaaaaaaaaaa}
\displaystyle{\times \, \,
\tau \Bigl(
n + P^{-} \! +\! 1, \bar n +1, {\bf t} + [a_{s}^{-1}]+ \sum_{k=1}^{P^{-}}
[b_{k}^{-1}] \Bigr) }
\\ \\ + \,
\displaystyle{
\sum_{s = 1}^{P^{-}} \prod_{i = 1, \neq s}^{P^{-}} E^{-1}(b_{s}, b_{i}) \,
\prod_{k = 1}^{P^{+}} E(b_{s}, a_{k})\,
\tau \Bigl( 
n + P^{-} \! -\! 1, \bar n, {\bf t} + \sum_{i\neq s}^{P^{-}}
[b_{i}^{-1}] \Bigr)}
\\ \\ \phantom{aaaaaaaaaaaaaaaaaaaaaaaaa}
\displaystyle{\times \, 
\tau \Bigl(
n + P^{+} \! +\! 1, \bar n +1, {\bf t} + [b_{s}^{-1}] + 
\sum_{k = 1}^{P^{+}}
[a_{k}^{-1}] \Bigr)}
\\ \\
\displaystyle{
=\, \tau \Bigl ( n+P^+ , \bar n +1, {\bf t} + \sum_{i=1}^{P^{+}}
[a_i^{-1}]\Bigr )
\tau \Bigl ( n+P^- , \bar n, {\bf t} + \sum_{k=1}^{P^{-}}
[b_k^{-1}]\Bigr )
}
\end{array}
\eeq
where 
\beq\label{E}
E(z, \zeta )=z^{-1} - \zeta^{-1}.
\eeq
Note that $P^{+} - P^{-}\in 2\ZZ +1$ here.
Using the method developed in \cite{Shigyo13} for the KP and mKP
hierarchies, one can prove
the following statement:

\begin{proposition}
Equations (\ref{bil3}) and (\ref{bil4}) (assuming that the latter holds
for all $a_i, b_k$ and $P^{+} - P^{-}\in 2\ZZ +1$) are equivalent.
\end{proposition}

The simplest non-trivial particular cases of (\ref{bil4}) are:

\noindent
\underline{$(P^{+}, P^{-})=(3, 0)$}:
\beq\label{bil5}
\begin{array}{l}
\phantom{-}E^{-1}(a_1, a_2)E^{-1}(a_1, a_3)
\tau \Bigl( n + 2, \bar n, {\bf t} + [a_2^{-1}] + [a_3^{-1}]\Bigr)
\tau \Bigl( n+1, \bar n+1,{\bf t} + [a_1^{-1}]\Bigr)
\\ \\
+\, E^{-1}(a_2, a_1)E^{-1}(a_2, a_3)
\tau \Bigl( n + 2, \bar n, {\bf t} + [a_1^{-1}] + [a_3^{-1}]\Bigr)
\tau \Bigl( n+1, \bar n+1 ,{\bf t} + [a_2^{-1}]\Bigr)
\\ \\
+\,E^{-1}(a_3, a_1)E^{-1}(a_3, a_2)
\tau \Bigl( n + 2, \bar n, {\bf t} + [a_1^{-1}] + [a_2^{-1}]\Bigr)
\tau \Bigl( n+1, \bar n+1, {\bf t} + [a_3^{-1}]\Bigr)
\\ \\
=\,
\tau \Bigl( n + 3, \bar n+1, {\bf t} + [a_1^{-1}] + [a_2^{-1}] 
+ [a_3^{-1}]\Bigr)
\, \tau \Bigl( n, \bar n, {\bf t}\Bigr),
\end{array}
\eeq

\noindent
\underline{$(P^{+}, P^{-})=(2, 1)$}:
\beq\label{bil6}
\begin{array}{l}
\phantom{-}E^{-1}(a_1, a_2)E(a_1, b_1)
\tau \Bigl( n+1, \bar n, {\bf t} + [a_2^{-1}]\Bigr)
\tau \Bigl( n+2, \bar n+1, {\bf t} + [a_1^{-1}] + [b_1^{-1}]\Bigr)
\\ \\
+\, E^{-1}(a_2, a_1)E(a_2, b_1)
\tau \Bigl( n+1, \bar n, {\bf t} + [a_1^{-1}] \Bigr)
\tau \Bigl( n+2, \bar n+1, {\bf t} + [a_2^{-1}] + [b_1^{-1}]\Bigr)
\\ \\
+\,E(b_1, a_1)E(b_1, a_2)
\tau \Bigl( n, \bar n+1, {\bf t}\Bigr)
\tau \Bigl( n+3, \bar n, {\bf t} + [a_1^{-1}] + [a_2^{-1}]
+ [b_1^{-1}]\Bigr)
\\ \\
=\,  
\tau \Bigl( n + 2, \bar n+1, {\bf t} + [a_1^{-1}] + [a_2^{-1}]\Bigr)
\tau \Bigl( n+1, \bar n, {\bf t}+ [b_1^{-1}]\Bigr).
\end{array}
\eeq 
Following the analogy with the
KP and mKP hierarchies, we conjecture that these 
two equations are already equivalent
to the whole hierarchy given by (\ref{bil4}) or (\ref{bil3}).
Below this will be proved in the dispersionless limit.

\section{Dispersionless version of the modified DKP hierarchy}

According to \cite{TT95},
the passage to the dispersionless version of the 
hierarchy consists in introducing a dispersion parameter $\hbar$ and
the re-scaling
$t_{k}\to t_{k}/\hbar$ for all $k \geq 1$,
$n\to t_0/\hbar$,
$\bar n\to \bar t_0/\hbar$,
with the subsequent limit $\hbar \to 0$. More precisely,
instead of the tau-function we
introduce the function 
$F(t_0, \bar t_0, {\bf t}; \hbar )$ related to the tau-function 
by the formula
\beq\label{kp6}
\tau \Bigl (\hbar^{-1}t_0, \hbar^{-1}
\bar t_0,\hbar^{-1}{\bf t}  \Bigr )=
\exp \left ( \frac{1}{\hbar^2}\,
F(t_0, \bar t_0, {\bf t}; \hbar )\right )
\eeq
and consider the limit $\hbar \to 0$ (if it exists):
$$
F=F(t_0, \bar t_0, {\bf t})=
\lim\limits_{\hbar \to 0} 
F(t_0, \bar t_0, {\bf t}; \hbar ).
$$
The $F$-function satisfies 
an infinite number of highly nonlinear differential equations which are
obtained from the bilinear equations for the 
tau-function. 

\subsection{Algebraic formulation}
\label{section:algebraic1}

In the dispersionless limit, the variables $t_0, \bar t_0$ are
continuous, and difference equations in $n , \bar n$ turn into
differential equations in $t_0, \bar t_0$. To write them explicitly,
we introduce the differential operators
\beq\label{mmkp11}
\nabla (z)=\p_{0}+\sum_{k\geq 1}\frac{z^{-k}}{k}\, \p_{t_k},
\qquad
\p_{0}=\p_{ t_{0}}, \qquad
\bar \p_{0}=\p_{\bar t_{0}}.
\eeq
Then
\beq\label{mkp8}
\tau \Bigl (t_0/\hbar \pm 1, \bar t_0/\hbar \pm 1,
{\bf t}/\hbar \pm [z^{-1}]\Bigr )=
\exp \left (\frac{1}{\hbar^2}\, e^{\pm \hbar  \nabla (z)
\pm \hbar \bar \p_0 }F 
(t_0, \bar t_0, {\bf t}; \hbar )\right ),
\eeq
and the limit of (\ref{bil4}) is straightforward. 
Referring to \cite{SZ26a} for details, we present the result:
\beq\label{bil7}
\begin{array}{l}
\displaystyle{
\phantom{+}\sum_{s = 1}^{P^{+}}
\left(
\prod_{i = 1, \neq s}^{P^{+}} E^{-1}(a_{s}, a_{i})
e^{ -\nabla (a_{i}) \nabla (a_{s}) F}
\right )
\left(
\prod_{k=1}^{P^{-}} E(a_{s}, b_{k}) 
e^{ \nabla (a_{s})  \nabla(b_{k})F}
\right)e^{\nabla (a_s)\bar \p_0 F}}
\\ \\ \!  +\,\,
\displaystyle{
\sum_{s = 1}^{P^{-}}
\left(
\prod_{i = 1, \neq s}^{P^{-}} E^{-1}(b_{s}, b_{i})
e^{ -\nabla (b_{i}) \nabla (b_{s}) F}
\right )`
\left(
\prod_{k = 1}^{P^{+}} E(b_{s}, a_{k}) 
e^{ \nabla (b_{s})  \nabla(a_{k})F}
\right)e^{-\nabla (b_s)\bar \p_0 F}}
\\ \\
=\,
\displaystyle{
\left ( \prod_{i=1}^{P^{+}}
e^{\nabla (a_i)\bar \p_0 F}\right )
\left (\prod_{k=1}^{P^{-}}
e^{-\nabla (b_k)\bar \p_0 F}\right )
}
\end{array}
\eeq
(recall that $P^+ + P^- \in  2\ZZ +1$ here). In particular,
the limits of equations (\ref{bil5}), (\ref{bil6}) are as follows:

\noindent
\underline{$(P^{+}, P^{-})=(3, 0)$}:
\beq\label{bil5d}
\begin{array}{l}
\phantom{-}E(a_1, a_2)
e^{(\nabla (a_1)\nabla (a_2)+\nabla (a_3)\bar \p_0) F}
\\ \\
\phantom{aaaaaaaaaaa}
+ E(a_2, a_3)
e^{(\nabla (a_2)\nabla (a_3)+\nabla (a_1)\bar \p_0) F}
\\ \\
\phantom{aaaaaaaaaaaaaaaaaaaaaa}
+ E(a_3, a_1)
e^{(\nabla (a_3)\nabla (a_1)+\nabla (a_2)\bar \p_0) F}
\\ \\
=E(a_1, a_2)E(a_1, a_3)E(a_2, a_3)
e^{\nabla (a_1)\nabla (a_2)F +\nabla (a_2)\nabla (a_3)F+
\nabla (a_3)\nabla (a_1)F+ (\nabla (a_1)+ \nabla (a_2)+
\nabla (a_3)) \bar \p_0 F},
\end{array}
\eeq

\noindent
\underline{$(P^{+}, P^{-})=(2, 1)$}:
\beq\label{bil6d}
\begin{array}{l}
\phantom{-}E(a_1, b_1)
e^{(\nabla (a_1)\nabla (b_1)+\nabla (a_1)\bar \p_0 +
\nabla (b_1)\bar \p_0 )F} 
\\ \\
\phantom{aaaaaaa}+\, E(b_1, a_2)
e^{(\nabla (a_2)\nabla (b_1)+\nabla (a_2)\bar \p_0 +
\nabla (b_1)\bar \p_0 )F}
\\ \\
\phantom{aaaaaaaaaaaaaa}+\,E(a_2, a_1)
e^{(\nabla (a_1)\nabla (a_2)+\nabla (a_1)\bar \p_0 +
\nabla (a_2)\bar \p_0 )F}
\\ \\
=\, E(a_1, a_2)E(a_1, b_1)E(b_1, a_2) 
e^{(\nabla (a_1)\nabla (a_2) +\nabla (a_1)\nabla (b_1)+
\nabla (a_2)\nabla (b_1))F}.
\end{array}
\eeq 

To proceed, it is convenient to introduce the functions
\beq\label{gvw}
\begin{array}{l}
g(z, \zeta )=E(z, \zeta )e^{\nabla (z)\nabla (\zeta )F},
\\ \\
w(z)=z^{-1}e^{\nabla (z)\p_0F}=g(z, \infty ),
\\ \\
v(z)=e^{\nabla (z)\bar \p_0F}
\end{array}
\eeq
(recall that $E(z, \zeta )=z^{-1}-\zeta^{-1}=-E(\zeta , z)$).
In this notation, equations (\ref{bil7}), (\ref{bil5d}), (\ref{bil6d})
acquire the form
\beq\label{bil7a}
\begin{array}{l}
\displaystyle{
\phantom{+}\sum_{s = 1}^{P^{+}}
\left(
\prod_{i = 1, \neq s}^{P^{+}} g^{-1}(a_s, a_i)
\right )
\left(
\prod_{k=1}^{P^{-}} g(a_s, b_k)
\right) v(a_s)}
\\ \\ 
\phantom{aaaaaaa}+\,
\displaystyle{
\sum_{s = 1}^{P^{-}}
\left(
\prod_{i = 1, \neq s}^{P^{-}} 
g^{-1}(b_s, b_i)
\right )
\left(
\prod_{k = 1}^{P^{+}} g(b_s, a_k)
\right)v^{-1}(b_s)}
\\ \\
\phantom{aaaaaaaaaaaaaa}=\,
\displaystyle{
\prod_{i=1}^{P^{+}}
v(a_i)
\prod_{k=1}^{P^{-}}
v^{-1}(b_k),
}
\end{array}
\eeq

\beq\label{bil5a}
\begin{array}{l}
g(a_1, a_2)v(a_3)+g(a_2, a_3)v(a_1)+g(a_3, a_1)v(a_2)
\\ \\
\phantom{aaaaaaaaaaaaa}=
g(a_1, a_2)g(a_1, a_3)g(a_2, a_3)v(a_1)v(a_2)v(a_3),
\end{array}
\eeq

\beq\label{bil6a}
\begin{array}{l}
g(a_1, a_2)v(a_1)v(a_2)+g(a_2, a_3)v(a_2)v(a_3)+
g(a_3, a_1)v(a_1)v(a_3)
\\ \\
\phantom{aaaaaaaaaaaaaaaaaaaaaaa}=g(a_1, a_2)g(a_1, a_3)g(a_2, a_3).
\end{array}
\eeq

\noindent
Letting $a_3\to \infty$ in (\ref{bil5a}), (\ref{bil6a}) and denoting
$a_1=z$, $a_2=\zeta$, we obtain the following relations:
\beq\label{bil6c}
\left \{ \begin{array}{l}
R g(z, \zeta )+v(z)w(\zeta )=w(z)v(\zeta )\Bigl (
1+R v(z)w(\zeta )g(z, \zeta )\Bigr ),
\\ \\
v(\zeta ) \Bigl (R w(\zeta )+ v(z)g(z, \zeta )\Bigr )=
w(z) \Bigl (Rv(z) +w(\zeta )g(z, \zeta )\Bigr ),
\end{array} \right.
\eeq
where
\beq\label{R}
R=v (\infty )=e^{\p_0 \bar \p_0 F}.
\eeq

\begin{proposition} 
The functions $g(z, \zeta )$, $v(z)$ satisfy the equation
\beq\label{bil8}
v^2(\zeta )\Bigl (g^2(z, \zeta )v^2 (z) +1\Bigr )-
\Bigl (g^2(z, \zeta )+v^2 (z)\Bigr ) +V(\zeta )g(z, \zeta )
v(z)=0,
\eeq
where
\beq\label{V(z)}
V(\zeta )=w^{-1}(\zeta )\Bigl (R^{-1} v^2 (\zeta ) -R\Bigr )
+w(\zeta )\Bigl (R v^2 (\zeta ) -R^{-1}\Bigr )
\eeq
and $R$ is given in (\ref{R}).
\end{proposition}

\noindent
{\it Proof.} Rewriting equations (\ref{bil5a}), (\ref{bil6a})
in the form
$$
\begin{array}{l}
g(a_1, a_2)v(a_3) \Bigl ( g(a_1, a_3)g(a_2, a_3)v(a_1)v(a_2)-1 \Bigr )
\\ \\
\phantom{aaaaaaaaaaaaaaaaaa}
=\, g(a_2, a_3) v(a_1)+g(a_3, a_1) v(a_2),
\\ \\
g(a_1, a_2) \Bigl (g(a_1, a_3)g(a_2, a_3)-v(a_1)v(a_2)\Bigr )
\\ \\
\phantom{aaaaaaaaaaaaaaaaaa}
=\, v(a_3)\Bigl (g(a_2, a_3) v(a_2)+g(a_3, a_1) v(a_1)\Bigr ),
\end{array}
$$
we can exclude $g(a_1, a_2)$ by dividing them. After some
transformations, this results in the relation
$$
\begin{array}{l}
\displaystyle{
\phantom{=}v^2(a_3) 
\left (g(a_2, a_3)v(a_2) +\frac{1}{g(a_2, a_3)v(a_2)}\right )
-\frac{v(a_2)}{g(a_2, a_3)}-\frac{g(a_2, a_3)}{v(a_2)}}
\\ \\
\displaystyle{
=\,
v^2(a_3) \left (g(a_1, a_3)v(a_1) +\frac{1}{g(a_1, a_3)v(a_1)}\right )
-\frac{v(a_1)}{g(a_1, a_3)}-\frac{g(a_1, a_3)}{v(a_1)}}
\end{array}
$$
which means separation of variables: the left-hand side 
is a function of $a_2$ while the right-hand side is the same
function of $a_1$. Therefore, we conclude that 
the combination
$$
v^2(\zeta ) 
\left (g(z, \zeta )v(z) +\frac{1}{g(z, \zeta )v(z)}\right )
-\frac{v(z)}{g(z, \zeta )}-\frac{g(z, \zeta )}{v(z)}:=-V(\zeta )
$$
does not depend on $z$. To find $V(\zeta )$, we let $z$ tend to
$\infty$. This gives (\ref{bil8}) and (\ref{V(z)}).
\square

\begin{corollary}
The functions $w(z)$, $v(z)$ satisfy the equation
\beq\label{bil8a}
R^2 \Bigl (w^2(z)v^2 (z) +1\Bigr )-
\Bigl (w^2(z)+v^2 (z)\Bigr ) +Vw(z)v(\zeta )=0,
\eeq
where $R=v(\infty )$ and
\beq\label{V}
V=V(\infty )=2\, e^{\p_0 (\bar \p_0 -\p_0)F}\p_{t_1}\bar \p_0 F
\eeq
do not depend on $z$.
\end{corollary}

\noindent
{\it Proof.}
This equation is the limiting case of (\ref{bil8}) as
$\zeta \to \infty$. The particular value of $V$ (\ref{V}) follows
from (\ref{V(z)}) in this limit.
\square

Equation (\ref{bil8}) defines an elliptic curve ${\cal E}$ 
together
with a marked point on it. More precisely, a point $P\in {\cal E}$
of the curve is a pair of functions 
$(v(z), g(z, \zeta ))$ satisfying equation
(\ref{bil8}), and $z^{-1}$ is a local parameter in a neighborhood
of infinity on the curve. 
Position of the marked point is parametrized
by the variable $\zeta$. The following proposition shows that
the elliptic modulus of the curve is indeed the same 
for all $\zeta$.

\begin{proposition}
The function 
$\displaystyle{
v^2(\zeta )+v^{-2}(\zeta )-\frac{V^2(\zeta )}{4 v^2(\zeta )}}
$
does not depend on $\zeta$. In particular, letting $\zeta$ 
tend to $\infty$,
we have:
\beq\label{VV}
v^2(\zeta )+v^{-2}(\zeta )-\frac{V^2(\zeta )}{4 v^2(\zeta )}
=R^2 +R^{-2}-
\frac{V^2}{4R^2}.
\eeq
\end{proposition}

\noindent
{\it Proof.}
This equality can be verified by a direct calculation which uses
expression (\ref{V(z)}) for $V(\zeta )$ and equation (\ref{bil8a})
in the form
\beq\label{VVV}
V=\frac{v^2(\zeta )+w^2(\zeta )-R^2 -R^2v^2(\zeta )
w^2(\zeta )}{v(\zeta )\, w(\zeta )}.
\eeq
\square

\noindent
Both the elliptic modulus and the marked point are dynamical
variables: they depend on the times $t_0, \bar t_0, {\bf t}$.

Letting $a_3\to \infty$ in equations (\ref{bil5a}),
(\ref{bil6a}) and denoting
$a_1=z$, $a_2=\zeta$, we can express 
$g(z, \zeta )=(z^{-1}-\zeta^{-1})e^{\nabla (z)\nabla (\zeta )F}$ 
through
the functions $w(z)$, $v(z)$ in two ways:
\beq\label{g1}
g(z, \zeta )=
\frac{w(z)v(\zeta )-w(\zeta )v(z)}{R\Bigl (1-
w(z)w(\zeta )v(z)v(\zeta )\Bigr )},
\eeq
as it follows from (\ref{bil5a}) or
\beq\label{g1a}
g(z, \zeta ) =\frac{R\, (w(\zeta )v(\zeta )-w(z )v(z ))}{w(z )w(\zeta )-
v(z )v(\zeta )},
\eeq
as it follows from (\ref{bil6a}).
The two expressions are equivalent because
the functions $w,v$ are constrained by equation (\ref{bil8a}).
They are the main generating equations of the hierarchy. Expanding,
say, (\ref{g1}) in inverse powers of $z, \zeta$, 
one obtains various differential
equations containing a finite number of derivatives with respect
to different times. This feature is common
for all dispersionless hierarchies: general 
mixed second order derivatives  
$\p_{t_n}\p_{t_m}F$ (of which a generating function is the
left-hand side) are expressed through some special second order
derivatives contained in the functions $v(z), w(z)$, i.e.,
$\p_{t_k}\p_0F$, $\p_{t_k}\bar \p_0F$ and $\p_0 \bar \p_0F$. 
Values of these latter derivatives can be regarded as Cauchy data
for equations of the hierarchy.
However, in contrast to the dispersionless KP hierarchy,
they are not independent since they are constrained by
the equation of the elliptic curve.
Equating the right-hand sides of (\ref{g1}), (\ref{g1a}),
we get the following relation:
\beq\label{R2}
R^2 = \frac{\Bigl (w(z)v(\zeta )-w(\zeta )v(z)\Bigr )\Bigl (v(z)v(\zeta )-
w(z)w(\zeta )\Bigr )}{\Bigl (w(z)v(z)-w(\zeta )v(\zeta )\Bigr )
\Bigl (1-w(z)w(\zeta )
v(z)v(\zeta )\Bigr )},
\eeq
which is used later.

Equations (\ref{g1}), (\ref{g1a}) are also remarkable 
for the following reason.

\begin{theorem}\label{theorem:equivalence}
Let the functions $w, v$ be constrained by equation 
of the elliptic curve (\ref{bil8a}), then
equation (\ref{g1}) (or (\ref{g1a}))
is equivalent to the whole hierarchy
(\ref{bil7}).
\end{theorem}

\noindent
To prove this, the elliptic parametrization of the Pfaff-type hierarchies
first suggested in \cite{AZ14,AZ15} seems to be unavoidable. It consists
in a special change of dynamical variables based
on (and motivated by) uniformization of the elliptic 
curve defined by equation
(\ref{bil8a}) by means of elliptic functions, which is considered
in Section \ref{section:elliptic}.

\subsection{Yang-Baxter formulation}
\label{section:yang-baxter}

In this section we suggest a new formulation of the 
dispersionless hierarchy
which has the form of the Yang-Baxter equation for a $4\times 4$
$R$-matrix of Baxter's form (the one associated with the
8-vertex model).

Let $w(z)$, $v(z)$, $g(z, \zeta )$ be the functions given by
(\ref{gvw}). Introduce also the following function:
\beq\label{f1}
f(z, \zeta )=v(\zeta )\, \frac{Rw(\zeta )+
v(z)g(z, \zeta )}{v(z)w(\zeta ) +R g(z, \zeta )},
\eeq
where $R$ is as in (\ref{R}). From equations (\ref{bil6c}) we see
that the function $f$ has another (equivalent) definition:
\beq\label{f1a}
f(z, \zeta )=v^{-1}(\zeta )\, \frac{Rv(z )+
w(\zeta )g(z, \zeta )}{1+ R v(z)w(\zeta )g(z, \zeta )}.
\eeq
Note that
\beq\label{f2}
f(z ,\infty )=v(z), \qquad f(z , z)=R.
\eeq

\begin{remark}
Recalling that $g(z ,\infty )=w(z)$, we may informally say that
the pair of functions $(f, g)$ is an analog of the pair $(v,w)$
(which encodes the Cauchy data)
``dressed'' by flows with respect to higher times of the hierarchy.
\end{remark}

\noindent
We will also need the 
function $f(\infty , z)$. In the limit
$\zeta \to \infty$ in $f(\zeta , z)$ both numerator and denominator
in (\ref{f1a}) tend to zero. Resolving the indeterminacy, we get:
\beq\label{f2a}
f(\infty , z)=v(z)\, \frac{p(z) +v_0}{p(z)-v_0},
\eeq
where
\beq\label{f2b}
\left.
p(z)=\p_{\zeta^{-1}}\Bigl (\frac{g(\zeta , z)}{w(z)}\Bigr )
\right |_{\zeta^{-1}\to 0}= z -\nabla (z)\p_{t_1}F,
\eeq
\beq{\label{v0}
v_0=-\p_{\zeta^{-1}} \log v(\zeta )\Bigr |_{\zeta^{-1}\to 0}=
-\bar \p_0 \p_{t_1}}F.
\eeq

Let us combine the functions $f$, $g$ into the 
following $R$-matrix of Baxter's form:
\beq\label{R0}
{\sf R}(z, \zeta )  =\left (
\begin{array}{cccc}
f(z, \zeta ) & 0 & 0& R g(z, \zeta )f(z, \zeta )
\\ \\
0 & g(z, \zeta ) & R & 0
\\ \\
0 & R & g(z, \zeta ) & 0
\\ \\
 R g(z, \zeta )f(z, \zeta ) & 0 & 0&  f(z, \zeta )
\end{array}
\right ) .
\eeq
In particular,
\beq\label{R0a}
{\sf R}(z, \infty )  =\left (
\begin{array}{cccc}
v(z) & 0 & 0& R w(z)v(z)
\\ \\
0 & w(z)  & R & 0
\\ \\
0 & R & w(z) & 0
\\ \\
 R w(z)v(z) & 0 & 0&  v(z)
\end{array}
\right ) 
\eeq
and
\beq\label{P1a}
{\sf R}(z, z )=R \, {\sf P},
\eeq
where ${\sf P}$ is the permutation matrix (\ref{P1}).
Since $R$, $f$ and $g$ are functions of the times, the $R$-matrix
depends on the times, too: ${\sf R}(z, \zeta )=
{\sf R}(z, \zeta ; t_0, \bar t_0, {\bf t})$.
As a rule, we will not indicate the dependence of times explicitly.

\begin{theorem}\label{theorem:yang-baxter}
The Yang-Baxter equation for the $R$-matrix 
(\ref{R0}), i.e.,
\beq\label{YB1}
{\sf R}_{12}(z_1, z_2) {\sf R}_{13}(z_1, z_3) {\sf R}_{23}(z_2, z_3)=
{\sf R}_{23}(z_2, z_3) {\sf R}_{13}(z_1, z_3) {\sf R}_{12}(z_1, z_2),
\eeq
is equivalent to the dispersionless modified DKP hierarchy
(\ref{bil7}).
\end{theorem}

\noindent
{\it Proof.}
The proof is based on the analysis of the Yang-Baxter equation
first made by Baxter himself for his solution of the 8-vertex
model. We outline his arguments in the next section using the
more familiar notation commonly adopted in the analysis 
of the matrix of Boltzmann weights in the 8-vertex model.

Here we only present the result, translating it from the language
of vertex models and rewriting in the notation introduced above.
Namely, equation (\ref{curve1}) that follows from the 
Yang-Baxter equation in our case implies that
\beq\label{bil8b}
R^2 \Bigl (f^2(z, \zeta )g^2(z, \zeta )+1 \Bigr )
-f^2(z, \zeta )-g^2(z, \zeta )+ V f(z, \zeta )g(z, \zeta )=0,
\eeq
where 
$V$ is a constant (which may depend 
on the times of the hierarchy). Putting here $\zeta =\infty$, 
we obtain the equation
\beq\label{YB3}
R^2 \Bigl (v^2(z)w^2(z)+1 \Bigr )
-v^2(z)-w^2(z)+ V v(z)w(z)=0,
\eeq
which is identical to (\ref{bil8a}). From it we can express 
the constant $V$
as
\beq\label{YB4}
V=\frac{w^2(z)+v^2(z)-R^2 -R^2 w^2(z)v^2(z)}{w(z)v(z)}
\eeq
(although the functions in the right-hand side are functions
of $z$, the 
Yang-Baxter equation (\ref{YB1}) implies that the whole
expression does not depend
on it!). Plugging $f$ from (\ref{f1}) and $V$ from (\ref{YB4})
into (\ref{bil8b}), we obtain, after some transformations:
\beq\label{YB5}
v^2 (\zeta )\Bigl (g^2 (z, \zeta ) v^2 (z)+1\Bigr )-
g^2 (z, \zeta ) - v^2 (z) +V(\zeta )g (z, \zeta )v (z)=0,
\eeq
where
\beq\label{V(z)a}
V(\zeta )=\frac{R^2 \Bigl (w^2 (\zeta )v^2(\zeta )-1\Bigr ) +v^2(\zeta )-
w^2(\zeta )}{R\, w(\zeta )}.
\eeq
It is easy to see that $V(\infty )=V$ given by (\ref{YB4}), and
equation (\ref{YB5}) at $\zeta =\infty$ coincides with (\ref{YB3}).

Subtracting equation (\ref{YB5}) from the same equation with
the change $z\leftrightarrow \zeta$, we get:
\beq\label{g100}
g(z, \zeta )=\frac{2\, (v^2(z)-v^2(\zeta ))}{V(z)v(\zeta )+
V(\zeta )v(z)}.
\eeq
Plugging here $V(z)$, $V(\zeta )$ from (\ref{V(z)a}) with
$R^2$ from (\ref{R2}), we obtain, after some transformations, that
$g(z, \zeta )$ is given by the same expression
as in (\ref{g1a}) (or in (\ref{g1})).
\square

\begin{proposition}
The $R$-matrix (\ref{R0}) satisfies the unitarity condition, 
i.e., it holds
\beq\label{uni}
{\sf R}_{12}(z, \zeta ){\sf R}_{21}(\zeta , z) =F(z, \zeta )\,{\sf I},
\eeq
where 
\beq\label{uni1}
F(z, \zeta )=R^2 - g^2(z, \zeta )
\eeq
and
${\sf I}$ is the $4\! \times \! 4$ identity matrix.
\end{proposition}

\noindent
{\it Proof.} It is enough to show that
\beq\label{uni2}
f(z, \zeta )f(\zeta , z)-R^2g^2(z, \zeta )f(z, \zeta )f(\zeta , z)
=R^2 - g^2(z, \zeta )=F(z, \zeta ).
\eeq
Writing equations (\ref{bil8b}) for the pairs $(z, \zeta )$
and $(\zeta , z)$ and subtracting them (taking into account that 
$g(z, \zeta )=-g(\zeta , z)$), we get:
$$
\Bigl (f^2(z, \zeta )-f^2 (\zeta , z)\Bigr )\Bigl (R^2 g^2(z, \zeta )
-1 \Bigr )+Vg(z, \zeta )\Bigl (f(z, \zeta )+f (\zeta , z)\Bigr )=0,
$$
from which we obtain
$$
f (\zeta , z)=f (z, \zeta )+\frac{V\, 
g(z, \zeta )}{R^2 g^2(z, \zeta )-1}.
$$
Plugging this into (\ref{uni2}), we can see that it is
equivalent to (\ref{bil8b}).
\square

\begin{remark}
It might seem that the unitarity of the $R$-matrix 
already implies the Yang-Baxter equation for it (because the both
properties are based on the same identities for elliptic functions).
However, this is not the case: there are many solutions to the
unitarity condition (\ref{uni}) that do not solve the Yang-Baxter
equation. Only the opposite is true: all solutions to the Yang-Baxter
equation of the special form (\ref{R10a}) satisfy the unitarity
condition.
\end{remark}

\begin{remark}
There is also another important property of the Baxter $R$-matrix
called {\it crossing unitarity} (see, e.g., \cite{Zotov15}).
It can be shown that for the $R$-matrix (\ref{R0}) the crossing unitarity 
is equivalent to unitarity of the matrix $\tilde R_{12}(z, \zeta )=
\sigma_2^{(1)}R_{12}(z, \zeta )\sigma_2^{(1)}$:
\beq\label{unia}
\tilde {\sf R}_{12}(z, \zeta )
\tilde {\sf R}_{21}(\zeta , z) =F(z, \zeta )\,{\sf I},
\eeq
with the same function $F(z, \zeta )$ as in (\ref{uni1}).
\end{remark}

Equation (\ref{bil8b}) means that the pair of functions
$(f, \, g)$ represents a point on the elliptic curve
${\cal E}= {\cal E}(t_0, \bar t_0, {\bf t})$ defined by
equations (\ref{bil8}) or (\ref{bil8a}). The same curve
can be defined by yet another equation which is obtained by
letting $z\to \infty$ in (\ref{bil8b}) and using
(\ref{f2a}).

\begin{proposition}
The elliptic curve ${\cal E}$ defined by
equations (\ref{bil8}) or (\ref{bil8a}) is isomorphic to the
curve defined by the equation
\beq\label{bil8c}
R_0^2 \Bigl (w^2 (z)+w^{-2}(z)\Bigr )=p^2(z)+V_0,
\eeq
where 
\beq\label{R0b}
R^2_0=\frac{4R^2 v_0^2}{V^2}, \qquad
V_0=R_0^2\left (R^2 +R^{-2} -
\frac{V^2}{4R^2}\right )
\eeq
($v_0$ is given in (\ref{v0})).
\end{proposition}

\noindent
{\it Proof.}
Plugging (\ref{f2a}) into (\ref{bil8}) and using (\ref{bil8a}),
we get, after cancellations:
$$
2v_0 \Bigl (w^2(z)-R^2\Bigr )=V \Bigl (p(z)+v_0 \Bigr )w(z)v(z).
$$
This equation allows us to express $v(z)$ through $w(z)$ and
$p(z)$. The substitution into (\ref{bil8a}) yields (\ref{bil8c}).
\square

\begin{remark}
It follows from (\ref{R0b}) that
$$
\frac{V_0}{R_0^2}=R^2 +R^{-2} -
\frac{V^2}{4R^2}=k+k^{-1},
$$
which means that the elliptic modulus $k$ is indeed the same
for the both curves.
\end{remark}

\section{Baxter's $R$-matrix and dispersionless hierarchies}
\label{section:elliptic}

\subsection{Baxter's solutions to the Yang-Baxter equation}
\label{section:baxter}

In its most general form the Yang-Baxter equation reads
\beq\label{RRR1a}
{\sf R}_{12}{\sf R}'_{13} {\sf R}''_{23}=
{\sf R}''_{23} {\sf R}'_{13} {\sf R}_{12}.
\eeq
All solutions to this equation among 
$4\! \times \! 4$ matrices were found in \cite{LP24}.
Here we are interested in the class of solutions having the
special form (\ref{R10a}). They correspond to matrices of Boltzmann
weights for the symmetric 
8-vertex (or 6-vertex) models. In this case, all solutions 
were found by Baxter. In the case of general position all the 4 Boltzmann weights $a,b,c,d$ 
are not identically zero. In this paper we 
consider this general case.
For the 6-vertex model the weight $d$ is identically zero. 
This case is simpler by itself, but the limiting procedure 
from the general case is somewhat tricky and
requires a special consideration to be
presented in the forthcoming paper \cite{SZ-to-appear}.

In the general case, when the fourth weight $d$ is not identically zero,
the solutions are described by the following theorem.

\begin{theorem} \cite{Baxter}
\label{theorem:Baxter}
Let ${\sf R}\in {\rm End} (V\otimes V)$ be a matrix of the form
$$
{\sf R}=\sum_{j=0}^3 w_j \sigma_j \otimes \sigma_j,
$$
where
$$
\begin{array}{l}
w_0=\frac{1}{2}(a+b), \quad w_3=\frac{1}{2}(a-b), 
\quad w_1=\frac{1}{2}(c+d), \quad w_2=\frac{1}{2}(c-d)
\end{array}
$$
(with the assumption that all $w_i$'s are nonzero),
i.e., in the basis $\{ {\sf e}_1\otimes {\sf e}_2, \;
{\sf e}_1\otimes {\sf e}_2, \; {\sf e}_2\otimes {\sf e}_1, \; 
{\sf e}_2 \otimes {\sf e}_2 \}$ of the space $V\otimes V$ the matrix is
\beq\label{R1}
{\sf R}=\left (
\begin{array}{cccc}
a & 0 & 0& d 
\\ 
0 & b & c & 0
\\ 
0 & c & b & 0
\\ 
d & 0 & 0& a
\end{array}
\right ) ,
\eeq
where $a\neq \pm b$ and $c\neq \pm d$. Assume also that 
all $a,b,c,d\neq 0$.
Further, let ${\sf R}'$, ${\sf R}''$ be matrices of the same structure
with nonzero matrix elements $a', b', c', d'$ and $a'', b'', c'', d''$
respectively. Then the equation
\beq\label{RRR1}
{\sf R}_{12}{\sf R}'_{13} {\sf R}''_{23}=
{\sf R}''_{23} {\sf R}'_{13} {\sf R}_{12}
\eeq
holds if and only if the quantities
\beq\label{RRR2}
\begin{array}{l}
\displaystyle{
\Gamma = \frac{cd}{ab},
\qquad
\Delta = \frac{a^2 +b^2 -c^2 -d^2}{ab}}
\end{array}
\eeq
are the same for all the three matrices, i.e., 
$\Gamma =\Gamma' =\Gamma ''$, $\Delta =\Delta '=\Delta''$.
\end{theorem}

\noindent
{\it Proof.} We follow original Baxter's proof given in \cite{Baxter}.
A direct calculation shows that equation (\ref{RRR1}) is equivalent
to the following system of 6 equations:
\beq\label{system1}
\left \{ \begin{array}{l}
ac'a'' +d a' d'' = bc' b'' + c a' c'',
\\ \\
ab'c'' +d d' b'' = ba' c'' + c c' b'',
\\ \\
cb'a'' +b d' d'' = ca' b'' + b c' c'',
\\ \\
ad'b'' +d b' c'' = bd' a'' + c b' d'',
\\ \\
aa'd'' +d c' a'' = bb' d'' + c d' a'',
\\ \\
da'a'' +a c' d'' = db' b'' + a d' c''.
\end{array}
\right.
\eeq
These are linear homogeneous equations with respect to
$a'', b'' , c'', d''$. A non-trivial solution exists if
the determinant combined from coefficients of, say, first, third,
fourth and sixth equations vanishes. The calculation shows that,
taking into account the assumptions, vanishing of the determinant
is equivalent to the relation
\beq\label{RRR3}
\frac{cd}{ab}=\frac{c'd'}{a'b'},
\eeq
i.e., $\Gamma =\Gamma '$. Imposing this condition, we can resolve
the equations with respect to $a'', b'', c'', d''$ up to a common
multiplier:
$$
\begin{array}{l}
\displaystyle{
a''= \frac{a}{c}(cc' -dd')(b^2 c'{}^{2} -c^2 a'{}^{2})},
\\ \\
\displaystyle{
b''= \frac{b}{d}(dc' -cd')(a^2 c'{}^{2} -d^2 a'{}^{2})},
\\ \\
\displaystyle{
c''= \frac{c}{a}(bb' -aa')(a^2 c'{}^{2} -d^2 a'{}^{2})},
\\ \\
\displaystyle{
d''= \frac{d}{b}(ab' -ba')(b^2 c'{}^{2} -c^2 a'{}^{2})}.
\end{array}
$$
Plugging this to the second equation in (\ref{system1}), we get,
taking into account (\ref{RRR3}):
\beq\label{RRR4}
\frac{a^2 +b^2 -c^2 -d^2}{ab}=
\frac{a'{}^{2} +b'{}^{2} -c'{}^{2} -d'{}^{2}}{a' b'},
\eeq
i.e., $\Delta =\Delta '$. 

Next, the system of equations (\ref{system1}) remains the same
if one interchanges the variables with prime and two primes, hence 
$\Gamma =\Gamma' =\Gamma ''$, $\Delta =\Delta '=\Delta''$.
\square

Assuming that $d\neq 0$,
we can put $d=abc$, then (\ref{RRR3}) yields
$\Gamma =c^2$, i.e., we should assume that the matrix
elements $c$ are the same for the three matrices: $c=c'=c''$.
Then from the second equation in (\ref{RRR2}) we have:
$$
\Delta =\frac{a^2 +b^2 -c^2 (1+a^2 b^2)}{ab},
$$
or
\beq\label{curve1}
\Gamma^2 (1+a^2 b^2)- (a^2 +b^2 )+\Delta ab =0.
\eeq
For $\Gamma \neq 0$ and $\Delta \neq \pm 2\Gamma^2 \pm 2$ 
this equation defines a smooth elliptic curve ${\cal E}$: 
a point $P=(a,b)\in {\cal E}$, if $a,b$ satisfy equation (\ref{curve1}).
Uniformization of this curve by means of elliptic functions
gives the well known Baxter's $R$-matrix with a spectral parameter
on the elliptic curve. The spectral parameter is just the uniformizing
parameter on the curve.

Namely, in terms of the elliptic function\footnote{This function
is the slightly modified ``elliptic sinus'' function, see Appendix B.} 
\beq\label{ell1}
{\sf sn}(u)={\sf sn}(u|\tau )=
\frac{\theta_1(u|\tau )}{\theta_4(u|\tau )},
\eeq
where $\theta_i(u|\tau )$ are Jacobi's theta-functions with the
modular parameter $\tau$ the uniformization is as follows:
\beq\label{ell2}
a=a(u)={\sf sn}(u +\eta ), \quad b=b(u)={\sf sn}(u),
\eeq
with
\beq\label{ell2a}
\Gamma ={\sf sn}(\eta ), \quad 
\Delta =2\, \frac{ \theta_4^2(0)\, 
\theta_2 (\eta )\theta_3 (\eta )}{\theta_2 (0 )\theta_3 (0 )\,
\theta_4^2 (\eta )},
\eeq
where $\eta$ is a parameter (which can be regarded
as a fixed marked point on the curve, namely
the point $P_0=(\Gamma , 0)\in {\cal E}$), i.e., from
the two independent parameters $\Gamma , \Delta$ 
we have passed to $\tau , \eta$. The third independent parameter
$u$ is the spectral parameter.
In this notation, the matrix elements of the $R$-matrix (\ref{R1})
which is now regarded as a matrix-valued function of the spectral
parameter, read:
\beq\label{RRR6a}
a={\sf sn}(u +\eta ), \quad b={\sf sn}(u),
\quad c={\sf sn}(\eta ), \quad d={\sf sn}(\eta ){\sf sn}(u)
{\sf sn}(u +\eta ).
\eeq
This is the well known elliptic parametrization of Baxter's 
$R$-matrix\footnote{This $R$-matrix is also known 
in other normalizations, which differ by a common factor. 
For example, in a number of problems, normalization is convenient, in which the Boltzmann weights are entire functions of $u$.}.
The parameters $\tau , \eta$ are the same for the three $R$-matrices
in (\ref{RRR1}) while the spectral parameters are different.
A simple argument shows that they are connected as 
$u'=u + u''$.
In this parametrization equation (\ref{RRR1}) acquires the
form of the Yang-Baxter equation for the $R$-matrix depending
on the spectral parameter:
\beq\label{RRR6}
{\sf R}_{12}(u_1-u_2) {\sf R}_{13}(u_1-u_3) {\sf R}_{23}(u_2-u_3)=
{\sf R}_{23}(u_2-u_3) {\sf R}_{13}(u_1-u_3) {\sf R}_{12}(u_1-u_2).
\eeq
The substitution (\ref{RRR6a}) converts it into identity.

It remains to comment on the two limiting cases of Baxter's $R$-matrix
when elliptic functions degenerate into trigonometric (or hyperbolic)
ones. They are:
\beq\label{two}
\begin{array}{l}
\mbox{I)} \; \tau \to +i0 ,
\\ \\
\mbox{II)} \; \tau \to +i\infty .
\end{array}
\eeq
In case I) the procedure is more or less straightforward after the
modular transformation $\tau \to \tau '=-1/\tau$ and subsequent limit
$\tau '\to +i\infty$. The Boltzmann weights $a,b,c,d$ remain 
nonzero but they are parametrized by trigonometric
(or hyperbolic) functions instead of elliptic ones. 
Case II) corresponds to the 6-vertex model ($d=0$).  
It seems to be easier since the limit $\tau \to +i\infty$
does not require any modular transformation, but definition of
the spectral curve in this limit requires some additional care.
The details will be presented in the forthcoming paper 
\cite{SZ-to-appear}.

\subsection{Elliptic parametrization of the dispersionless
mDKP hierarchy}
\label{section:elliptic1}

The key step in the proof of Theorem \ref{theorem:equivalence}
is a special change of variables motivated by uniformization
of the elliptic curve (\ref{bil8a}). Namely, instead of 
the original dynamical variables which are combined into
the generating functions $w(z), \, v(z)$, we introduce another
set of dependent variables $c_k=c_k(t_0, \bar t_0, {\bf t})$
which are coefficients in the series
\beq\label{u(z)}
u(z)=c_1 z^{-1} +\sum_{k\geq 2}c_k z^{-k},
\eeq
where it is assumed that $z$ belongs to a neighborhood of $\infty$.
The function $u(z)$ is normalized so that $u(\infty )=0$.
Then, setting
\beq\label{unif1}
w(z)={\sf sn}(u(z)), \qquad 
v(z)={\sf sn}(u(z)+\eta )
\eeq
and
\beq\label{unif2}
R={\sf sn}(\eta ), \qquad
V=2\frac{\theta_4^2(0|\tau )
\theta_2(\eta |\tau )\theta_3(\eta|\tau  )}{\theta_2(0|\tau )
\theta_3(0|\tau )\theta_4^2(\eta |\tau )}=\frac{2\, {\sf cn}(\eta ) \,
{\sf dn}(\eta )}{{\sf cn} (0) \,
{\sf dn}(0)},
\eeq
where ${\sf sn (u)}$ and the other 
elliptic functions are defined in (\ref{uni3a}), we see that 
equation (\ref{bil8a}) is satisfied identically for any value
of the modular parameter $\tau$ (which is not indicated
explicitly in (\ref{unif1})). In fact this is the same 
uniformization as given above for Baxter's $R$-matrix
(see (\ref{ell2}), (\ref{ell2a})):
one should identify $w(u)$ with $b(u)$,
$R$ with $\Gamma$, $V$ with $\Delta$.
In (\ref{unif1}), there is a new times-dependent variable $\eta$
which represents a marked point on the elliptic curve.
The modular parameter $\tau$ is a dynamical variable, too.
It is determined by the equation
\beq\label{proof3}
R^2 +R^{-2} -\frac{V^2}{4R^2}=
\Bigl (\frac{{\sf cn}(0)}{{\sf dn}(0)}\Bigr )^2 +
\Bigl (\frac{{\sf dn}(0)}{{\sf cn}(0)}\Bigr )^{2}
\eeq
(see Appendix B for details). The main advantage of this change
of variables is that instead of two functions $v(z), w(z)$ constrained
by the equation of the curve we now deal with only one, $u(z)$.

\begin{proposition}\label{proposition:fg}
In the elliptic parametrization, the functions $g(z, \zeta )$ and
$f(z, \zeta )$ are given by
\beq\label{fg}
\begin{array}{l}
g(z, \zeta )={\sf sn} \Bigl (u(z)-u(\zeta )\Bigr ), 
\\ \\
f(z, \zeta )={\sf sn} \Bigl (u(z)-u(\zeta )+\eta \Bigr ).
\end{array}
\eeq
\end{proposition}

\noindent
{\it Proof.}
The assertion follows (after applying some well known
identities for elliptic functions) from substitution of 
(\ref{unif1}), (\ref{unif2})
into (\ref{g1}) and (\ref{f1}).
\square

\begin{remark}
Comparing with (\ref{ell2}), we see that $f$, $g$ and $R$ are 
analogous to the Boltzmann weights $a$, $b$ and $c$ respectively.
\end{remark}

Proposition \ref{proposition:fg} allows us to prove Theorem
\ref{theorem:equivalence}. Indeed, after the substitutions
(\ref{fg}), (\ref{unif1}) the general Hirota-Miwa 
equation (\ref{bil7a}) 
acquires the form
\beq\label{bil7b}
\begin{array}{l}
\displaystyle{
\phantom{+}\sum_{s = 1}^{P^{+}}
\left(
\prod_{i = 1, \neq s}^{P^{+}} {\sf sn}^{-1}(u_s, u_i)
\right )
\left(
\prod_{k=1}^{P^{-}} {\sf sn}(a_s, v_k)
\right) {\sf sn}(u_s+\eta )}
\\ \\ 
\phantom{aaaaaaa}+\,
\displaystyle{
\sum_{s = 1}^{P^{-}}
\left(
\prod_{i = 1, \neq s}^{P^{-}} 
{\sf sn}^{-1}(v_s, v_i)
\right )
\left(
\prod_{k = 1}^{P^{+}} {\sf sn}(v_s, u_k)
\right){\sf sn}^{-1}(v_s+\eta )}
\\ \\
\phantom{aaaaaaaaaaaaaa}=\,
\displaystyle{
\prod_{i=1}^{P^{+}}
{\sf sn}(u_i +\eta )
\prod_{k=1}^{P^{-}}
{\sf sn}^{-1}(v_k+\eta ),
}
\end{array}
\eeq
where $u_i=u(a_i)$, $v_k=u(b_k)$. In the same way as it was done
in \cite{SZ26a}, one can prove that this equation is satisfied
identically for all the variables $u_i, v_k$. The proof is outlined
in Appendix A.

Therefore, we conclude that the whole hierarchy is equivalent
to the equations
\beq\label{hi1}
\left \{
\begin{array}{rcl}
(z^{-1}-\zeta^{-1})e^{\nabla (z)\nabla (\zeta )F}&=&
{\sf sn} \Bigl (u(z)-u(\zeta )\Bigr )
\\ && \\
e^{\nabla (z)\bar \p_0 F}& = & {\sf sn}\Bigl (u(z)+\eta \Bigr ).
\end{array} \right.
\eeq
Note that the elliptic parametrization 
of the function $w(z)$ (\ref{unif1}) is already contained
in the first of these equations in the limit $\zeta \to \infty$.

For completeness, it remains to extend 
the elliptic parametrization to the 
curve (\ref{bil8c}).

\begin{proposition}
The substitutions
\beq\label{sub1}
\begin{array}{l}
\displaystyle{
p(z)=\gamma  \, 
\frac{{\sf cn}(u(z))\,{\sf dn}(u(z))}{{\sf sn}(u(z))},}
\\ \\
R_0^2 = \gamma^2 {\sf cn}^2(0)\, {\sf dn}^2(0),
\\ \\
V_0 =\gamma^2 \Bigl ({\sf cn}^4(0)+ {\sf dn}^4(0)\Bigr ),
\end{array}
\eeq
where $\gamma =\gamma (t_0, \bar t_0, {\bf t})$ is a
dynamical variable, convert equation
(\ref{bil8c}) into identity.
\end{proposition}

\noindent
{\it Proof.} This identity can be proved directly but it is
probably more instructive to derive the parametrization
(\ref{sub1}) from the one for the functions $f$, $g$ (\ref{fg})
using (\ref{f2a}) and (\ref{R0b}). 
From (\ref{f2a}) we have:
$$
p(z)=v_0 \, \frac{f(\infty , z)+v(z)}{f(\infty , z)-v(z)}.
$$
From the definition of $v_0$ (\ref{v0}) we get:
$$
v_0=-\p_{\zeta^{-1}}\log {\sf sn}\Bigl (u(\zeta )+\eta \Bigr )
\Bigr |_{\zeta^{-1}=0}=
-\p_{\zeta^{-1}}u(\zeta )
\Bigr |_{\zeta^{-1}=0} \, \frac{{\sf sn}'(\eta )}{{\sf sn}(\eta )}
$$
$$
=\, -\pi c_1 \theta^2_4(0) \, 
\frac{{\sf cn} (\eta )\, {\sf dn} (\eta )}{{\sf sn} (\eta )},
$$
where we have used equation (\ref{ini5}) from Appendix B and
$c_1$ is the coefficient at the first term in (\ref{u(z)}).
Next, we have:
$$
\begin{array}{lll}
\displaystyle{
- \frac{f(\infty , z)+v(z)}{f(\infty , z)-v(z)} }& =&
\displaystyle{
\frac{{\sf sn}(u+\eta )-
{\sf sn}(u-\eta )}{{\sf sn}(u+\eta )+{\sf sn}(u-\eta )}}
\\ && \\
&=&\displaystyle{
\frac{\theta_1(u(z)+\eta |\tau )\theta_4(u(z)-\eta |\tau )-
\theta_1(u(z)-\eta |\tau )
\theta_4(u(z)+\eta |\tau )}{\theta_1(u(z)+\eta |\tau )
\theta_4(u(z)-\eta |\tau )+
\theta_1(u(z)-\eta |\tau )\theta_4(u(z)+\eta |\tau )}}
\\ && \\
&=&\displaystyle{\frac{\theta_2 (u(z)|\frac{\tau}{2})\,
\theta_1 (\eta |\frac{\tau}{2})}{\theta_1 (u(z)|\frac{\tau}{2})\,
\theta_2 (\eta |\frac{\tau}{2})}}
\\ && \\
&=&\displaystyle{\frac{\theta_2 (u(z)|\tau )\, \theta_2 (u(z)|\tau )\,
\theta_1 (\eta |\tau )\, \theta_4 (\eta |\tau )}{\theta_1 (u(z)|\tau )\,
\theta_4 (u(z)|\tau )\,
\theta_2 (\eta |\tau )\, \theta_3 (\eta |\tau )}=
\frac{{\sf cn} (u(z))\, {\sf dn} 
(u(z))\, {\sf sn} (\eta )}{{\sf cn} (\eta )\, {\sf dn} 
(\eta )\, {\sf sn} (u(z))}}.
\end{array}
$$
Combining this with the expression for $v_0$ obtained above,
we arrive at (\ref{sub1}) with
$\gamma =\pi \theta_4^2(0)c_1$.
\square

\section{Conclusion and outlook}

The main result of this work is  
establishing a close relationship between some 
integrable hierarchies of nonlinear differential equations 
with  zero dispersion and the Yang-Baxter equation 
for quantum R-matrices. More precisely, the main statement
is as follows. Consider the $4\times 4$ Baxter
$R$-matrices which are matrices of Boltzmann's weights 
$a,b,c,d$ for the 8-vertex or 6-vertex models and 
let's make them functions of the hierarchical times $t_k$ 
by expressing them in terms of mixed second partial derivatives 
of the $F$-function with respect to the times. Then the
Yang-Baxter equation for such $R$-matrix is equivalent 
to some integrable hierarchy in the limit of zero dispersion
(in general, if $d\neq 0$, to the dispersionless modified 
DKP hierarchy).

Our proof of the equivalence is based on identification
of the elliptic curves that are 
part of the algebraic structure of the two integrable systems.
On the Yang-Baxter side, it is the spectral curve on which the
spectral parameter lives. On the side of dispersionless 
hierarchies, it is the dynamical curve. However, there is a
significant difference between them. The modular parameter
$\tau$ of the spectral curve is a fixed constant, which is a
parameter of the model, while the dynamical curve
``breathes'', that is, its modular parameter 
is generally a time-dependent dynamical variable.
Translating this into the language of vertex models, 
their physical quantities  (say, partition functions) 
can be made time-dependent, according to a solution of the 
dispersionless hierarchy. It is then natural to ask: whether 
the partition functions (or some other quantities specific to
vertex models or quantum spin chains), constructed in this way
as functions of the times, satisfy some nice 
differential equations.
This is an interesting problem for further research.

In this paper we have considered only the general case, when
the elliptic curve is smooth. 
Two possible trigonometric degenerations of the elliptic 
Baxter $R$-matrix deserve a special consideration. Presumably,
they correspond to the dispersionless modified large BKP and mKP
hierarchies. Also,
the Baxter $R$-matrix admits further degeneration, namely, to
the Yang $R$-matrix ${\sf R}(u)=u{\sf I}+{\sf P}$ with
rational dependence on the spectral parameter. The 
natural question arises whether any 
integrable hierarchies can be associated with rational $R$-matrices.
A similar question can be asked regarding other possible
solutions to the Yang-Baxter equation, not necessarily of the form
(\ref{R10a}), for example, the non-symmetric 
$R$-matrices corresponding to 
the 7-vertex or 11-vertex models (see, e.g., \cite{AHZ97,LOZ14} and
references therein). We hope to discuss all these 
degenerate cases in a
separate publication \cite{SZ-to-appear}.

Another possible subject for future research is an extension
of the connection with $R$-matrices to solutions
to the Yang-Baxter equation among matrices of a higher size,
for example, to the Belavin $R$-matrices of size $n^2\times n^2$ 
for which the spectral curve is still elliptic \cite{Belavin81}.
The chiral Potts model \cite{chiralPotts}, 
where the spectral curve is a curve
of higher genus, may be also of interest in this respect.
In this regard, it would also be interesting to understand 
what a dispersionless hierarchy turns into at the classical limit 
of the quantum $R$-matrix, that is, whether (and in what form) any
relationship of integrable hierarchies with the classical version 
of the Yang-Baxter equation for the classical $r$-matrix 
will be preserved.

From the side of integrable hierarchies, a possible generalization
might be an extension of this approach to the full family
of multi-component dispersionless hierarchies. 
The goal here (however, hardly achievable in such a formulation) 
could be a recipe for how to assign to each dispersionless integrable 
hierarchy a certain quantum $R$-matrix 
satisfying the Yang-Baxter equation. 

Finally, the most ambitious related goal 
(again, hardly achievable in the short term) is 
to extend the relationship found in this work 
to the ``full'' hierarchies (KP, DKP, BKP) 
with nonzero dispersion. Currently, there are not 
even any meaningful hypotheses about how and with what 
it is necessary to replace quantum $R$-matrices 
and the Yang-Baxter equation.

\section*{Appendix A: Theta-functions and elliptic functions}
\label{appendix:theta}

\addcontentsline{toc}{section}{Appendix A: Theta-functions
and elliptic functions}
\def\theequation{A\arabic{equation}}
\def\theHequation{\theequation}
\setcounter{equation}{0}

\subsection*{Theta-functions}

The Jacobi's theta-functions $\theta_a (u)=
\theta_a (u|\tau )$, $a=1,2,3,4$, are defined by the absolutely 
convergent infinite sums as follows:
\beq\label{Bp1}
\begin{array}{l}
\theta _1(u)=-\displaystyle{\sum _{k\in \z}}
\exp \left (
\pi i \tau (k+{\textstyle\frac{1}{2}})^2 +2\pi i
(u+{\textstyle\frac{1}{2}})(k+{\textstyle\frac{1}{2}})\right ),
\\ \\
\theta _2(u)=\displaystyle{\sum _{k\in \z}}
\exp \left (
\pi i \tau (k+{\textstyle\frac{1}{2}})^2 +2\pi i
u(k+{\textstyle\frac{1}{2}})\right ),
\\ \\
\theta _3(u)=\displaystyle{\sum _{k\in \z}}
\exp \left (
\pi i \tau k^2 +2\pi i u k \right ),
\\ \\
\theta _4(u)=\displaystyle{\sum _{k\in \z}}
\exp \left (
\pi i \tau k^2 +2\pi i
(u+{\textstyle\frac{1}{2}})k\right ),
\end{array}
\eeq where $\tau$ is a complex parameter (the modular parameter) 
such that $\Im \, \tau >0$. The function 
$\theta_1(u)$ is odd, the other three functions are even.
The infinite product representation for the theta-functions reads: 
\beq\label{infprod}
\begin{array}{ll}
\theta_1(u|\tau)&=\displaystyle{2q^{\frac{1}{4}}\sin\pi u
\prod_{n=1}^\infty(1-q^{2n})(1-q^{2n}e^{2\pi i u})(1-q^{2n}e^{-2\pi i u}),}
\\ & \\
\theta_2(u|\tau)&=\displaystyle{2q^{\frac{1}{4}}\cos\pi u
\prod_{n=1}^\infty(1-q^{2n})(1+q^{2n}e^{2\pi i u})(1+q^{2n}
e^{-2\pi i u}),}
\\ & \\
\theta_3(u|\tau)&=\displaystyle{
\prod_{n=1}^\infty(1-q^{2n})(1+q^{2n-1}e^{2\pi i u})(1+q^{2n-1}
e^{-2\pi i u}),}
\\ & \\
\theta_4(u|\tau)&=\displaystyle{
\prod_{n=1}^\infty(1-q^{2n})(1-q^{2n-1}e^{2\pi i u})
(1-q^{2n-1}e^{-2\pi i u})}.
\end{array}
\eeq
where $q=e^{\pi i \tau}$.
In the limit $\tau \to +i\infty$ they are:
$\theta_1(u|\tau )=2q^{\frac{1}{4}}\sin \pi u + O(q^{\frac{9}{4}})$,
$\theta_2(u|\tau )=2q^{\frac{1}{4}}\cos \pi u + O(q^{\frac{9}{4}})$,
$\theta_3(u|\tau )=1+O(q)$, $\theta_4(u|\tau )=1+O(q)$.
The theta-functions satisfy a lot of nontrivial identities.

The main transformation properties of the theta functions are
listed below.

\noindent
\underline{Shifts by (quasi)periods}:
\beq\label{p}
\begin{array}{l}
\theta_1(u+1)=-\theta_1(u),\\ \\
\theta_2(u+1)=-\theta_2(u),\\ \\
\theta_3(u+1)=\theta_3(u),\\ \\
\theta_4(u+1)=\theta_4(u).
\end{array}\hspace{2cm}
\begin{array}{l}
\theta_1(u+\tau)=-e^{-\pi i(2u+\tau)}\theta_1(u),\\ \\
\theta_2(u+\tau)=e^{-\pi i(2u+\tau)}\theta_2(u),\\ \\
\theta_3(u+\tau)=e^{-\pi i(2u+\tau)}\theta_3(u),\\ \\
\theta_4(u+\tau)=-e^{-\pi i(2u+\tau)}\theta_4(u).
\end{array}\vspace{0.3cm}
\eeq
\underline{Shifts by half-periods}:
\beq\label{hp}
\begin{array}{l}
\theta_1(u+{\textstyle\frac{1}{2}})=\theta_2(u),\\ \\
\theta_2(u+{\textstyle\frac{1}{2}})=-\theta_1(u),\\ \\
\theta_3(u+{\textstyle\frac{1}{2}})=\theta_4(u),\\ \\
\theta_4(u+{\textstyle\frac{1}{2}})=\theta_3(u).
\end{array}\hspace{2cm}
\begin{array}{l}
\theta_1(u+{\textstyle\frac{\tau}{2}})=
ie^{-\pi i(u+\tau/4)}\theta_4(u),\\ \\
\theta_2(u+{\textstyle\frac{\tau}{2}})=
e^{-\pi i(u+\tau/4)}\theta_3(u),\\ \\
\theta_3(u+{\textstyle\frac{\tau}{2}})=
e^{-\pi i(u+\tau/4)}\theta_2(u),\\ \\
\theta_4(u+{\textstyle\frac{\tau}{2}})=i
e^{-\pi i(u+\tau/4)}\theta_1(u).
\end{array}\vspace{0.3cm}
\eeq
\underline{The modular transformation $\tau \to -1/\tau$}:
\beq\label{mod2d}
\begin{array}{ll}
\theta_1\left (u/\tau|-1/\tau\right)&=-i\sqrt{-i\tau}\,
e^{\pi iu^2/\tau}\theta_1(u|\tau),\\ & \\
\theta_2\left (u/\tau|-1/\tau\right)&=\sqrt{-i\tau}\,
e^{\pi iu^2/\tau}\theta_4(u|\tau), \\ & \\
\theta_3\left (u/\tau|-1/\tau\right)&=\sqrt{-i\tau}\,
e^{\pi iu^2/\tau}\theta_3(u|\tau),\\ & \\
\theta_4\left (u/\tau|-1/\tau\right)&=\sqrt{-i\tau}\,
e^{\pi iu^2/\tau}\theta_2(u|\tau).
\end{array}
\eeq
(The branch of the square root here
is such that $\Re \sqrt{-i\tau}>0$.)

The theta-functions satisfy a lot of nontrivial identities. 
The ones that are often used in the present work are listed below:
\beq\label{identities}
\begin{array}{l}
\theta_4^2(u)\theta_4^2(v)-\theta_1^2(u)\theta_1^2(v)
=\theta_4^2(0)\theta_4(u+v)\theta_4(u-v),
\\ \\
\theta_1^2(u)\theta_4^2(v)-\theta_4^2(u)\theta_1^2(v)
=\theta_4^2(0)\theta_1(u+v)\theta_1(u-v),
\\ \\
\theta_1'(0)=\pi \theta_2(0)\theta_3(0)\theta_4(0).
\end{array}
\eeq
All theta-functions here have the same modular parameter $\tau$.
A more complete list of identities can be found in \cite{KZ15}.

\subsection*{Elliptic functions}

Here we give the necessary information on the elliptic functions
in a very brief form. For the general theory of elliptic 
functions and theta-functions see
\cite{WW,Akhiezer,Takebe-book}.

Elliptic functions are double-periodic functions 
in the complex plane. By the Liouville theorem, non-constant
elliptic functions must necessarily have 
singularities in the fundamental domain. In the simplest case, 
such functions have two simple poles with opposite residues. 
They can be explicitly represented as ratios of the theta-functions.
The standard such functions are ${\rm sn}(w)$
(the ``elliptic sinus''), ${\rm cn}(w)$
(the ``elliptic cosinus'') and ${\rm dn}(w)$.
They are expressed through the 
Jacobi theta-functions 
in the following way:
\beq\label{uni3}
\begin{array}{l}
\displaystyle{
{\rm sn}(w)=\frac{\theta_3(0)\, 
\theta_1(u)}{\theta_2(0)\, \theta_4(u)}},
\quad
\displaystyle{
{\rm cn}(w)=\frac{\theta_4(0)\, 
\theta_2(u)}{\theta_2(0)\, \theta_4(u)}},
\quad
\displaystyle{
{\rm dn}(w)=\frac{\theta_4(0)\, 
\theta_3(u)}{\theta_3(0)\, \theta_4(u)}},
\end{array}
\eeq
where 
\beq\label{uni5}
u=\frac{w}{\pi \, \theta_3^2(0)},
\eeq
where the modular parameter $\tau$ is not shown explicitly.
The number
\beq\label{uni4}
k =\frac{\theta_2^2(0|\tau )}{\theta_3^2(0|\tau )}.
\eeq
is called the elliptic modulus. 

An advantage of the definition (\ref{uni3}) is that in the limit
$\tau \to +i\infty$
the functions
${\rm sn}$ and ${\rm cn}$ become the standard trigonometric 
functions $\sin$ and $\cos$ while the function ${\rm dn}$ converts
into the constant function equal to 1.
However, 
we found it more convenient to use in the main text the following
slightly modified functions:
\beq\label{uni3a}
\begin{array}{l}
\displaystyle{
{\sf sn}(u)=\frac{\theta_1(u)}{\theta_4(u)}},
\qquad
\displaystyle{
{\sf cn}(u)=\frac{\theta_2(u)}{\theta_4(u)}},
\qquad
\displaystyle{
{\sf dn}(u)=\frac{\theta_3(u)}{\theta_4(u)}},
\end{array}
\eeq
then
\beq\label{uni4a}
k =\Bigl (\frac{{\sf cn}(0)}{{\sf dn}(0)}\Bigr )^2.
\eeq
From (\ref{p}) and (\ref{hp}) we find their transformation
properties.

\noindent
\underline{Shifts by periods}:
\beq\label{p1}
\begin{array}{l}
\sn (u+1)=-\sn (u),\\ \\
\cn (u+1)=-\cn (u),\\ \\
\dn (u+1)=\dn (u),\\ \\
\end{array}\hspace{2cm}
\begin{array}{l}
\sn (u+\tau )=\sn (u),\\ \\
\cn (u+\tau )=-\cn (u),\\ \\
\dn (u+\tau )=-\dn (u),\\ \\
\end{array}\vspace{0.3cm}
\eeq
\underline{Shifts by half-periods}:
\beq\label{hp1}
\begin{array}{l}
\displaystyle{
\sn (u+{\textstyle\frac{1}{2}})=\; \frac{\cn (u)}{\dn (u)},}\\ \\
\displaystyle{
\cn (u+{\textstyle\frac{1}{2}})=-\frac{\sn (u)}{\dn (u)},}\\ \\
\displaystyle{
\dn (u+{\textstyle\frac{1}{2}})=\; \frac{1}{\dn (u)},}\\ \\
\end{array}\hspace{2cm}
\begin{array}{l}
\displaystyle{
\sn (u+{\textstyle\frac{\tau }{2}})=\; \frac{1}{\sn (u)},}\\ \\
\displaystyle{
\cn (u+{\textstyle\frac{\tau }{2}})=\; \frac{\dn (u)}{i\, \sn (u)},}\\ \\
\displaystyle{
\dn (u+{\textstyle\frac{\tau }{2}})=\; \frac{\cn (u)}{i\, \sn (u)}.}
\end{array}\vspace{0.3cm}
\eeq
We also need the formula for the $u$-derivative of the function
${\sf sn}(u)$:
\beq\label{ini5}
{\sf sn}'(u)=\pi \theta_4^2(0) \, {\sf cn}(u)\, {\sf dn}(u).
\eeq

At last, let us give some details of the proof of 
the identity (\ref{bil7c}). First of all, we change the variables
as $v_k \to v_k+\frac{\tau}{2}$ and $\eta_k \to \eta_k+\frac{\tau}{2}$
and rewrite the identity in the form
\beq\label{bil7c}
\begin{array}{l}
\displaystyle{
\phantom{+}\sum_{s = 1}^{P^{+}}
\left(
\prod_{i = 1, \neq s}^{P^{+}} {\sf sn}^{-1}(u_s, u_i)
\right )
\left(
\prod_{k=1}^{P^{-}} {\sf sn}^{-1}(a_s, v_k)
\right) {\sf sn}^{-1}(u_s+\eta )}
\\ \\ 
\phantom{aaaaaaa}+\,
\displaystyle{
\sum_{s = 1}^{P^{-}}
\left(
\prod_{i = 1, \neq s}^{P^{-}} 
{\sf sn}^{-1}(v_s, v_i)
\right )
\left(
\prod_{k = 1}^{P^{+}} {\sf sn}^{-1}(v_s, u_k)
\right){\sf sn}^{-1}(v_s+\eta )}
\\ \\
\phantom{aaaaaaaaaaaaaa}-\,
\displaystyle{
\prod_{i=1}^{P^{+}}
{\sf sn}^{-1}(u_i +\eta )
\prod_{k=1}^{P^{-}}
{\sf sn}^{-1}(v_k+\eta )=0.
}
\end{array}
\eeq
It is not difficult to see that the left-hand side
is proportional to sum of residues of the elliptic function
\beq\label{fff}
{\cal F}(u)=\sn^{-1}(u+\eta )\prod_{i=1}^{P^+}
\sn^{-1}(u-u_i) \prod_{k=1}^{P^-}
\sn^{-1}(u-v_k)
\eeq 
with periods $1$ and $\tau$ in the fundamental parallelogram,
and thus is equal to zero. (It is important here that
$P^+ +P^- \in 2\ZZ +1$, so the total number of
multipliers in (\ref{fff}) is even, 
otherwise $1$ is not a period but only
a half-period of this function.)

\section*{Appendix B: Uniformization of algebraic curves}
\label{appendix:uniformization}

\addcontentsline{toc}{section}{Appendix B: Uniformization of  
algebraic curves}
\def\theequation{B\arabic{equation}}
\def\theHequation{\theequation}
\setcounter{equation}{0}

Given a complex algebraic curve $\Gamma$ defined by an equation of the
form
\beq\label{P}
P(x,y)=0, \quad x,y \in \CC ,
\eeq
where $P(x,y)$ is a polynomial, a natural question is how this curve
can be uniformized. The uniformization means that there are two 
functions, $x(u)$ and $y(u)$, of some complex variable $u$
such that: 
\begin{itemize}
\item[a)] They are single-valued in a domain 
${\sf D}\in \CC$, 
\item[b)] The equation $P(x(u),y(u))=0$ is 
satisfied identically for all $u\in {\sf D}$, 
\item[c)] Any solution
to the equation $P(x,y)=0$ is obtained in this way 
for some $u\in {\sf D}$.
\end{itemize}

Smooth curves of degree higher than 2 can not be uniformized
by elementary functions. For example, uniformization of curves
defined by an equation of the form  $y^2 =Q(x)$, where $Q(x)$ is
a polynomial of degree 3 or 4, requires elliptic functions.
In the case of degree 4 the canonical form of the equation is
\beq\label{uni1a}
y^2 =(1-x^2)(1-k^2 x^2),
\eeq
where $k$ is a parameter (the elliptic modulus). If $k\neq 0,1$,
the curve is a smooth elliptic curve (a torus).
It can be uniformized by the elliptic functions ${\rm sn}(w)$
${\rm cn}(w)$ and ${\rm dn}(w)$ defined in (\ref{uni3}):
\beq\label{uni2a}
x(w)={\rm sn}(w), \quad y(w)={\rm cn}(w)\, {\rm dn}(w)=x'(w).
\eeq
The elliptic modulus $k$ is given by (\ref{uni4}).

The main example is the elliptic curve
\beq\label{B5}
R^2(x^2y^2+1)-(x^2+y^2)+Vxy=0
\eeq
which emerges in the solution of the Yang-Baxter equation
for $4\! \times \! 4$ matrices of Boltzmann weights for the 
8-vertex model.
The rational change of variables $(x,y)\to (X,Y)$, where
$$
Y=Rx^2y +\frac{Vx-2Ry}{2R}, \quad
X=k^{-1/2}x \quad \mbox{with 
$\displaystyle{k+k^{-1}=R^2 +R^{-2} -\frac{V^2}{4R^2}}$}
$$
brings equation (\ref{B5}) to the canonical form (\ref{uni1a}),
i.e., $Y^2 =(1-X^2)(1-k^2 X^2)$.
In the original variables the uniformization (\ref{uni2a})
acquires the form
\beq\label{B6}
x(u)=\frac{\theta_1(u|\tau )}{\theta_4(u |\tau )}={\sf sn}\, (u), \qquad
y(u)=\frac{\theta_1(u+\eta |\tau )}{\theta_4(u+\eta |\tau )}
={\sf sn}\, (u+\eta ).
\eeq
The two constants $R$, $V$ are expressed in terms 
of two parameters $\eta$, $\tau$ as follows:
\beq\label{B7}
R=\frac{\theta_1(\eta |\tau )}{\theta_4(\eta |\tau )}=
{\sf sn}\, (\eta ) , \qquad
V=2\frac{\theta_4^2(0|\tau )
\theta_2(\eta |\tau )\theta_3(\eta|\tau  )}{\theta_2(0|\tau )
\theta_3(0|\tau )\theta_4^2(\eta |\tau )}=\frac{2\, {\sf cn}(\eta ) \,
{\sf dn}(\eta )}{{\sf cn} (0) \,
{\sf dn}(0)}.
\eeq
The elliptic modulus $k$ is given by (\ref{uni4}) (or (\ref{uni4a})).

Another example is the curve
\beq\label{B1}
y^2 -R_0^2 (x^2 +x^{-2})+V_0=0,
\eeq
or, in the polynomial form,
\beq\label{B1a}
x^2y^2 -R_0^2x^4 +V_0x^2 -R^2=0.
\eeq
In this case the rational change of variables that brings
it to the canonical form is
$$
Y=\frac{xy}{R_0}, \quad
X=k^{-1/2}x \quad \mbox{with 
$\displaystyle{k+k^{-1}=\frac{V_0}{R_0^2}}$}.
$$
In terms of the original variables the uniformization 
(\ref{uni2a}) reads:
\beq\label{B2}
x(u)=\sn (u), \qquad
y(u)=\gamma \frac{\cn (u) \dn (u)}{\sn (u)}
\eeq
and
\beq\label{B3}
R_0=\gamma \,
\cn (u) \dn (u), \qquad V_0=\gamma^2
\Bigl (\cn^4(u)+\dn^4 (u)\Bigr ),
\eeq
where $\gamma$ is an arbitrary constant. 
To verify validity of these formulas directly, one should
prove the identity
\beq\label{B4}
\theta_4^4(0)\frac{\theta_2^2(u)\theta_3^2(u)}{\theta_1^2(u)
\theta_4^2(u)}-\theta_2^2(0)\theta_3^2(0)\left (
\frac{\theta_4^2(u)}{\theta_1^2(u)}+\frac{\theta_1^2(u)}{\theta_4^2(u)}
\right )+\theta_2^4(0)+\theta_3^4(0)=0.
\eeq
For the proof we note that the left-hand side is an even elliptic
function of $u$ with possible poles at $u=0$ and $u=\frac{\tau}{2}$.
However, the expansion around these points shows that the singular terms
cancel and the function is regular everywhere. This means that it is
a constant. To find the constant one can substitute any value of $u$.
It is convenient to take $u=\frac{1}{2}$. Using the transformation 
properties (\ref{hp}), one finds that the constant is zero.

\section*{Acknowledgments}
\addcontentsline{toc}{section}{Acknowledgments}

I am grateful to A. Savchenko for collaboration in \cite{SZ26a}
and A. Zotov for illuminating discussions on the Baxter $R$-matrix.

This article is an output of the research project 
(HSE-BR-2025-059) implemented as a part of the 
Basic Research Program at the National Research University 
Higher School of Economics (HSE University).


\begin{thebibliography}{99}

\addcontentsline{toc}{section}{References}

\bibitem{Baxter}
R. Baxter, {\it Exactly solved models in statistical mechanics},
Academic Press, 1982.

\bibitem{Baxter71}
R. Baxter, {\sl Eight-vertex model in lattice statistics},
Phys. Rev. Lett. {\bf 26} (1971) 832–-833.

\bibitem{Baxter72}
R. Baxter, {\sl Partition function of the eight-vertex lattice model},
Annals of Physics {\bf 70} (1972) 193–-228.

\bibitem{Baxter73} R. Baxter, {\sl Eight-Vertex Model in Lattice Statistics
and One-Dimensional Aniso\-tropic Heisenberg
Chain. I. Some Fundamental Eigenvectors,
II. Equivalence to a Generalized Ice-type Lattice Model,
III. Eigenvectors of the Transfer Matrix and Hamiltonian},
Annals of Physics {\bf 76} (1973) 1--24, 25--47, 48--71.

\bibitem{FT79} L. Takhtajan and L. Faddeev, {\sl The quantum method of
the inverse problem and the Heisenberg
$XYZ$ model}, % Uspekhi Mat. Nauk {\bf 34:5} (1979) 13--63 (English translation:
Russ. Math. Surveys {\bf 34:5} (1979) 11–-68.

\bibitem{Gaudin-book}
M. Gaudin, \textsl{La fonction d’onde de Bethe}, Masson, 1983.

\bibitem{KBI93} V.E. Korepin, N.M. Bogoliubov,
A.G. Izergin, {\sl Quantum Inverse Scattering Method and
Correlation Functions}, Cambridge: Cambridge Univ.
Press, 1993.

\bibitem{Slavnov-book}
N.A. Slavnov, {\it Algebraic Bethe Ansatz and Correlation Functions}, 
World Scientific, Singapore, 2022.


\bibitem{DJKM83} E. Date, M. Jimbo, M. Kashiwara and T. Miwa, {\it Transformation groups for soliton equations}, 
in: M. Jimbo, T. Miwa (Eds.), Nonlinear Integrable Systems --
Classical and Quantum, World Scientific, 1983, pp. 39--120.

\bibitem{JM83} M. Jimbo and T. Miwa, {\it Solitons and 
infinite dimensional Lie
algebras}, Publ. Res. Inst. Math. Sci. Kyoto {\bf 19} (1983) 943--1001.

\bibitem{HBbook} 
J. Harnad and F. Balogh, {\it Tau functions and their applications},
Cambridge Monographs on Mathematical Physics, Cambridge University
Press, 2021.

\bibitem{UT84} K. Ueno and K. Takasaki, {\it Toda lattice hierarchy},
Adv. Studies in Pure Math. {\bf 4} (1984) 1--95.



\bibitem{DJKM81} E. Date, M. Jimbo, M. Kashiwara and 
T. Miwa, {\it Transformation groups
for soliton equations III}, J. Phys. Soc. Japan {\bf 50} (1981) 3806--3812.

\bibitem{KL93} V. Kac and J. van de Leur, 
{\it The $n$-component KP hierarchy and 
representation theory}, in: A.S. Fokas, V.E. Zakharov (Eds.), 
Important Developments
in Soliton Theory, Springer-Verlag, Berlin, Heidelberg, 1993.



\bibitem{Teo11} L.-P. Teo, {\it The multicomponent KP hierarchy: 
differential Fay identities and Lax
equations}, J. Phys. A: Math. Theor. {\bf 44} (2011) 225201.

\bibitem{TT07} K. Takasaki and T. Takebe, {\it Universal Whitham hierarchy,
dispersionless Hirota equations and multicomponent KP hierarchy}, 
Physica D {\bf 235} (2007) 109--125.

\bibitem{TZ25}
T. Takebe, A. Zabrodin, {\it Multi-component Toda lattice hierarchy},
Russian Mathematical Surveys, {\bf 80:4} (2025)
591--665.







\bibitem{HO} R. Hirota and Y. Ohta, {\it Hierarchies of coupled 
soliton equations I}, J. Phys. Soc. Japan {\bf 60} (1991) 798-809.

\bibitem{AHM} M. Adler, E. Horozov and P. van Moerbeke, {\it 
The Pfaff lattice and skew-orthogonal polynomials}, Int. Math. Res.
Notices {\bf 1999} (1999), no 11, 569--588.

\bibitem{ASM} M. Adler, T. Shiota and P. van Moerbeke, {\it
Pfaff $\tau$-functions}, Math. Ann. {\bf 322} (2002) 423--476.


\bibitem{Kakei} S. Kakei, {\it Orthogonal and symplectic matrix integrals
and coupled KP hierarchy}, J. Phys. Soc. Japan {\bf 99} (1999) 2875--2877.



\bibitem{IWS} S. Isojima, R. Willox and J. Satsuma, {\it On various 
solutions of the coupled KP equation}, J. Phys. A: Math. Gen.
{\bf 35} (2002) 6893--6909.

\bibitem{Willox} R. Willox, {\it On a generalized Tzitzeica 
equation}, Glasgow Math. J. {\bf 47A} (2005) 221--231.




\bibitem{Kodama} Y. Kodama and K.-I. Maruno, {\it $N$-soliton 
solutions to the DKP hierarchy and the Weyl group actions}, 
J. Phys. A: Math. Gen. {\bf 39} (2006) 4063--4086.

\bibitem{Vandeleur} J. van de Leur, {\it 
Matrix integrals and the geometry of spinors},
J. Nonlinear Math. Phys. {\bf 8} (2001) 288-310.

\bibitem{Orlov}  A. Orlov, {\it 
Deformed Ginibre ensembles and integrable systems}, 
Phys. Lett. {\bf A 378} (2014) 319-328. 

\bibitem{SZ24}
A. Savchenko, A. Zabrodin, {\it Multicomponent DKP 
hierarchy and its dispersionless limit},
Lett. Math. Phys. {\bf 115:22} (2025).

\bibitem{SZ25a}
A. Savchenko, A. Zabrodin,
{\it Multi-component Pfaff-Toda hierarchy within bilinear formalism},
arXiv:2511.10779.




\bibitem{TT95}
K. Takasaki and T. Takebe, {\it Integrable hierarchies and dispersionless
limit}, Rev. Math. Phys. {\bf 7} (1995) 743--808. 

\bibitem{Takasaki07} K. Takasaki, {\it 
Differential Fay identities and auxiliary linear 
problem of integrable hiearchies},  
Advanced Studies in Pure Mathematics {\bf 61} (2011) 387--441.

\bibitem{Takasaki09} K. Takasaki, {\it Auxiliary linear problem,
difference Fay identities and dispersionless 
limit of Pfaff-Toda hierarchy},
SIGMA {\bf 5} (2009) 109.



\bibitem{AZ14} V. Akhmedova and A. Zabrodin, {\it Dispersionless
DKP hierarchy and elliptic L\"owner equation}, J. Phys. A: Math. Theor.
{\bf 47} (2014) 392001, arXiv:1404.5135.

\bibitem{AZ15}
V. Akhmedova, A. Zabrodin,
{\it Elliptic parametrization of Pfaff integrable hierarchies
in the zero dispersion limit}, 
Theor. Math. Phys. {\bf 185} (2015) 410--422,
arXiv:1412.8435.

\bibitem{SZ25b}
A. Savchenko, A. Zabrodin,
{\it Dispersionless version of multi-component Pfaff-Toda hierarchy},
Lett. Math. Phys. {\bf 116:51} (2026).

\bibitem{SZ26a}
A. Savchenko, A. Zabrodin, {\it 
Integrable hierarchies with zero dispersion and elliptic curves},
arXiv:2606.01354.

\bibitem{Z24}
A. Zabrodin, {\it Dispersionless version of 
the multicomponent KP hierarchy revisited},
Physica D {\bf 467} (2024) 134286,
arXiv:2404.10406. 








\bibitem{Miwa82}
T. Miwa, {\it On Hirota's difference equations}, Proc. Japan Acad. 
{\bf 58} Ser. A (1982) 9--12.

\bibitem{Shigyo13}
Y. Shigyo, {\it On addition formulae of KP, mKP
and BKP hierarchies}, SIGMA {\bf 9} (2013) 035.










\bibitem{LP24}
M. de Leeuw, V. Posch, {\it All $4\times 4$ solutions
of the quantum Yang-Baxter equation}, 
https://arxiv.org/abs/2411.18685v3.

\bibitem{SZ-to-appear}
A. Savchenko, A. Zabrodin, to appear.

\bibitem{KZ15}
S. Kharchev and A. Zabrodin, 
{\it Theta vocabulary I}, 
Journal of Geometry and Physics, {\bf 94} (2015) 19--31,
arXiv:1502.04603.

\bibitem{WW}
E.T. Whittaker, G.N. Watson, {\it A course of modern analysis},
Cambridge University Press, 1927.

\bibitem{Akhiezer}
N.I. Akhiezer, {\it Elements of the Theory of Elliptic 
Functions}, Translation of Mathematical Monographs, 
Vol. 79, (AMS, Providence RI, 1990). 

\bibitem{Takebe-book}
T. Takebe, {\it Elliptic Integrals and Elliptic Functions},
Moscow Lectures, volume 9, Springer, 2023.

\bibitem{AHZ97}
A. Antonov, K. Hasegawa, A. Zabrodin, {\it
On trigonometric intertwining vectors and non-dynamical $R$-matrix for
the Ruijsenaars model}, Nucl. Phys. {\bf B503} (1997) 747--770. 


\bibitem{LOZ14}
A. Levin, M. Olshanetsky, A. Zotov, {\it 
Classical integrable systems and soliton 
equations related to eleven-vertex $R$-matrix},
Nucl. Phys. {\bf B887} (2014) 400--422.

\bibitem{Belavin81}
A. Belavin, 
{\it Dynamical symmetry of integrable quantum systems},
Nucl. Phys. {\bf B180} (1981) 189–-200.

\bibitem{chiralPotts}
R. Baxter, J. Perk, H. Au-Yang, {\it New solutions of the star-triangle relations for the chiral potts model}, Phys. Lett.  
{\bf A128} (1988) 138–-142. 

\bibitem{Zotov15}
A.V. Zotov, {\it Associative Yang-Baxter equation and R-matrix identities: applications to integrable systems}, to appear.


\end{thebibliography}
\end{document}